\documentclass[sigconf,edbt]{acmart-edbt2018}
\usepackage{amsmath,epsfig}
\usepackage{mcomments}
\usepackage{pgfplots}
\usepackage{flexisym}
\pgfplotsset{width=9cm}

\makeatletter
\newcommand\resetstackedplots{
\makeatletter
\pgfplots@stacked@isfirstplottrue
\makeatother
\addplot [forget plot,draw=none] coordinates{(1,0) (2,0) (3,0) (4,0) (5,0) (6,0) (7,0) (8,0)};
}
\makeatother

\usepackage{graphicx}
\usepackage{balance}  

\usepackage{listings}

\usepackage{footmisc}

\usepackage{pgf,tikz}
\usepackage{amsmath}
\usetikzlibrary{arrows}
\usepackage{url}
\usepackage{booktabs}
\usepackage{balance}
\usepackage{paralist}
\lstdefinestyle{RDF}{basicstyle=\ttfamily,
                        keywordstyle=\lstuppercase,
                        emphstyle=\itshape,
                        showstringspaces=false,
                        }

\usepackage{graphicx}
\usepackage{amsmath}
\usepackage{caption,subcaption}
\usepackage{amsmath}
\usepackage{algorithm}
\usepackage{algpseudocode}

\usepackage{listings}
\usepackage{graphics}
\makeatletter
\newcommand{\lstuppercase}{\uppercase\expandafter{\expandafter\lst@token
                           \expandafter{\the\lst@token}}}
           
\makeatother
\newcommand{\jyoticomment}[1]{}

\lstdefinestyle{RDF}{basicstyle=\ttfamily,
                        keywordstyle=\lstuppercase,
                        emphstyle=\itshape,
                        showstringspaces=false,
                        }
\lstdefinelanguage[RDF]{SPARQL}[]{SPARQL}{
  morekeywords={SELECT, WHERE, FILTER, PREFIX},
}
\lstdefinestyle{customsparql}{
  belowcaptionskip=1\baselineskip,
  breaklines=true,
  xleftmargin=\parindent,
  language=SPARQL,
  showstringspaces=false,
  basicstyle=\scriptsize\ttfamily\color{blue},
  keywordstyle=\bfseries\color{green!40!black},
  commentstyle=\scriptsize\ttfamily\color{blue},
  identifierstyle=\color{blue},
  stringstyle=\color{blue},
}
\usepackage{balance}
\usepackage{verbatim}
\usepackage{pdfpages} 
\usepackage{amsmath}
\usepackage{xspace}

\newcommand{\streak}{\textsc{Streak}\xspace}

\newcommand{\squad}{{\textsf{S-QuadTree}}\xspace}

\usepackage{etoolbox}

\begin{document}

\title{STREAK: An Efficient Engine for Processing Top-$k$ SPARQL Queries with Spatial Filters}
\author{Jyoti Leeka} 
\affiliation{
	\institution{IIIT Delhi}
	\city{New Delhi} 
	\state{Delhi}
	\country{India}
}
\email{jyotil@iiitd.ac.in}

\author{Srikanta Bedathur} 
\affiliation{
	\institution{IBM Research}
	\city{New Delhi} 
	\state{Delhi}
	\country{India}
}
\email{sbedathur@in.ibm.com}

\author{Debajyoti Bera}
\affiliation{
	\institution{IIIT Delhi}
	\city{New Delhi} 
	\state{Delhi}
	\country{India}
}
\email{dbera@iiitd.ac.in}

\author{Sriram Lakshminarasimhan$~^{\dagger}$}\thanks{$^\dagger$now at Google, India}
\affiliation{
	\institution{IBM Research}
	\city{New Delhi} 
	\state{Delhi}
	\country{India}
}
\email{sriram.ls@gmail.com}

\begin{abstract} 
The importance of geo-spatial data in critical applications such as emergency response, transportation, agriculture etc., has prompted the adoption of recent GeoSPARQL standard in many RDF processing engines. In addition to large repositories of geo-spatial data --e.g., LinkedGeoData, OpenStreetMap, etc.-- spatial data is also routinely found in automatically constructed knowledge-bases such as Yago and WikiData. While there have been research efforts for efficient processing of spatial data in RDF/SPARQL, very little effort has gone into building end-to-end systems that can holistically handle complex SPARQL queries along with spatial filters.

In this paper, we present \textsc{Streak}, a RDF data management system that is designed to support a wide-range of queries with spatial filters including complex joins, top-$k$, higher-order relationships over spatially enriched databases. \textsc{Streak} introduces various novel features such as a careful identifier encoding strategy for spatial and non-spatial entities, the use of a semantics-aware Quad-tree index that allows for early-termination and a clever use of adaptive query processing with zero plan-switch cost. We show that \streak can scale to some of the largest publicly available semantic data resources such as Yago3 and LinkedGeoData which contain spatial entities and quantifiable predicates useful for result ranking. For experimental evaluations, we focus on \emph{top-k distance join queries} and demonstrate that \textsc{Streak} outperforms popular spatial join algorithms as well as state of the art end-to-end systems like Virtuoso and PostgreSQL.
\end{abstract}
\maketitle

\section{Introduction}
Motivated by its use in critical applications such as emergency response, transportation, agriculture etc., geo-spatial data are part of many large-scale semantic web resources. Efforts to standardise the representation of geo-spatial data and relationships between spatial objects within RDF/SPARQL has resulted in GeoSPARQL standard ~\cite{open2012ogc}, adopted by many RDF data repositories such as LinkedGeoData \cite{auer2009linkedgeodata}, GeoKnow~\cite{ios_geoknow_chapter}, etc. Even the large-scale open-domain knowledge-bases such as Yago~\cite{yago2,cidr_yago3}, WikiData~\cite{vrandevcic2014wikidata}, and DBPedia~\cite{lehmann2015dbpedia} contain significant amounts of spatial data. To make effective use of this rich data, it is crucial to efficiently evaluate queries combining topological and spatial operators --e.g., overlap, distance, etc.-- with traditional graph pattern queries of SPARQL. 

Furthermore, modern knowledge bases contain not only simple (binary) relationships between entities, but also model higher-order information such as uncertainty,  context, strength of relationships, and more. Querying such knowledge bases results in complex workloads involving user-defined ad-hoc ranking, with top-$k$ result cut-offs to help in reasoning under uncertainty and to reduce the resulting overload.  Thus the geo-spatial support is not just limited to efficient implementation of spatial operators that works well with SPARQL graph pattern queries, but also should integrate with top-$k$ early termination too. 

We present \streak, an efficient SPARQL query processing system that can model higher-order facts, support the modeling of geo-spatial objects and queries over them, enable efficient combined querying of graph pattern queries with spatial operators, expressed in the form of a \textsc{filter}, along with top-$k$ requirements over arbitrary user-defined ranking functions, expressed in the \textsc{order by --- limit} clauses. In this paper, we focus our attention on a recently proposed \emph{top-$k$ spatial distance join} (K-SDJ)~\cite{TopKSpatialJoins}, although other spatial operators are also possible within \streak. Following are a few applications of such distance join queries within the context of RDF databases from the GeoCLEF~2006-07~\cite{mandl2007geoclef}:
\begin{asparaenum}[(i)]
\item \textit{Find top wine producing regions around rivers in Europe}
\item \textit{List car bombings near Madrid and rank them by estimated casualties}
\item \textit{Find places near a flooded place which have low rates of house insurance}. 
\end{asparaenum}
Similarly, distance join queries over graph structured data can also be formulated in many diverse scenarios -- e.g., in a job search system one can ask for \textit{top jobs with highest wage which match the skills and are near the location of a given candidate}, which can be expressed as a SPARQL query with a spatial distance \textsc{filter} and top-$k$ ranking on highest wage.

\subsection{Challenge}
A large-body of work exists in relational database research for integrating geo-spatial and top-$k$ early termination operators within generic query processing framework. Therefore, a natural question to ask is whether those methods can be applied in a straightforward manner for RDF/SPARQL as well. Unfortunately, as we delineate below, the answer to this question is in the negative, primarily due to the schema-less nature of RDF data and the self-join heavy plans resulting from SPARQL queries. 

Consider applying one of the popular techniques, which build indexes over joining (spatial) tables during an offline/pre-processing phase ~\cite{transformers}. However, in RDF data, there is only one large triples table, resulting in a single large spatial index that needs to be used during many self-joins making it extremely inefficient. On the other hand, storing RDF in a property-table form is also impractical due to a large number of property-tables -- e.g., in a real-world dataset like BTC (Billion Triples Challenge)~\cite{btc}, there are $484,586$ logical tables~\cite{characteristicSets}.

Next, the unsorted ordering (w.r.t. the scoring function) of spatial attributes during top-$k$ processing prevents the straightforward application of methods that encode spatial entities to speed up spatial-joins~\cite{encodingSpatialRDF}. This is because these approaches assume sorted order of spatial attributes, and hence are able to choose efficient merge joins as much as possible. Finally, state-of-the-art query optimization techniques for top-$k$ queries~\cite{ilyas2004rank} cannot be adopted as-is for spatial top-$k$ since real-world spatial data does not follow uniform distribution assumptions that are used during query optimization. 

\subsection{Contributions}
In this paper, we present the design of the \streak system, and make following key contributions:
\begin{asparaenum}
	\item We present a novel soft-schema aware, \jyoticomment{in-memory} spatial index called \squad. It not only partitions the 2-d space to speed spatial operator evaluation, it compactly stores the inherent soft-schema captured using \textit{characteristic sets}. Experimental evaluations show that the use of \squad results up to \emph{three orders of magnitude} improvements over using state-of-the-art spatial indexes for many benchmark queries.
	\item We introduce a new spatial join algorithm, that utilizes the soft-schema stored in \squad to select optimal set of nodes in \squad. The spatial objects enclosed/overlapping the 2-d partitions of these nodes are used for pruning the search space during spatial joins.
	\item We propose an adaptive query plan generation algorithm called \textit{Adaptive Processing for Spatial filters} (\textbf{APS}), that switches with very low plan-switching overhead, to select between alternate query plans, based on a novel spatial join cost model.
\end{asparaenum}

We implement these features within \textsc{Quark-X}~\cite{quarkX}, a top-$k$ SPARQL processing system based on RDF-3X~\cite{neumann2008rdf,neumann2009scalable}, resulting in a holistic system capable of efficiently handling spatial joins. For experimental evaluation, we make use of two large real-world datasets --- (a) Yago~\cite{yago2,cidr_yago3}, and (b) LinkedGeoData~\cite{auer2009linkedgeodata}. Benchmark queries were focused on K-SDJ queries, and were a combination of subset of queries from GeoCLEF~\cite{mandl2007geoclef} challenge and queries we constructed using query features found important in literature~\cite{diversifiedStressTesting,quarkX,encodingSpatialRDF}. 
We compare the performance of \streak with the system obtained by replacing our \squad based spatial join by synchronous R-tree traversal spatial join algorithm~\cite{syncRTree}; and test end-to-end system performance against PostgreSQL~\cite{postgres} which has in-built spatial support and Virtuoso~\cite{virtuoso}, a state-of-the-art RDF/SPARQL processing system. Our experimental results show that \streak is able to outperform these systems for most benchmark queries by $1-2$ orders of magnitude in both cold and warm cache settings. 

\subsection{Organization}
The rest of the paper is organized as follows: Section~\ref{sec:prelim} briefly introduces the SPARQL query structure, and how spatial distance joins can be expressed. It also outlines a running example query that we use throughout the paper. Next, in Section~\ref{sec:ourFramework} we describe the key features of \streak framework including \squad, node selection algorithm, and APS algorithm. We describe the datasets and queries used in our evaluation in Section~\ref{sec:eval}. Detailed experimental results are given in Section~\ref{sec:experimentalResults} followed by an overview of related work in Section~\ref{sec:relwork} before concluding in Section~\ref{sec:concl}.

\section{Preliminaries}
\label{sec:prelim}
RDF consists of triples, where each triple represents relationship between \textit{subject(s)} and \textit{object(o)}, with name of the relationship being \textit{predicate(p)}. SPARQL is the standard query language for RDF~\cite{w3c2013sparql} and we consider SPARQL queries having the following structure: %
\begin{lstlisting}[deletekeywords=GRAPH]
SELECT [projection clause]
WHERE [graph pattern]
FILTER [spatial distance function]
ORDER BY [ranking function]
LIMIT [top-K results]
\end{lstlisting}
The \texttt{SELECT} clause retrieves results consisting of variables (denoted by ``?" prefix) and their bindings (for details see Figure\ref{fig:RDFKnowledgeGraph}). Graph patterns in WHERE clause (also called triple patterns) are expressed as: $\langle$\textsf{?s ?p ?o}$\rangle$, $\langle$\textsf{?r rdf:subject ?s. ?r rdf:predicate ?p. ?r rdf:object ?o}$\rangle$, where ?r represents fact or reification identifier. Keyword \textsf{prefix rdf} is part of the RDF vocabulary and helps declare parts of an RDF statement identified through its identifier \textit{?r}. 

Next, the \textrm{FILTER} clause restricts solutions to those that satisfy the spatial predicates. In this paper, we focus only on the use of DISTANCE predicate for distance joins. It should be noted that the techniques discussed in this paper are equally applicable to all spatial predicates defined in GeoSPARQL standard  by the Open Geospatial Consortium~\cite{parliament2012}. The ORDER BY clause helps establish the order of results using a user-defined ranking function.  The \textrm{LIMIT} clause  controls the number of results returned. Together, ORDER BY and LIMIT form the top-$k$ result retrieval with ad-hoc user-defined ranking function.

\subsection{Running Example} \label{sec:runex}
\begin{figure}[bt]
\centering
\footnotesize
\fontsize{7.1}{4}
\begin{verbatim}
@prefix xsd: <http://www.w3.org/2001/XMLSchema#>.	
#@ <id1>
:Mosel :grapeVariety :Albalonga.
#@ <id2>
:Mosel :soilType :porus_slate.
#@ <id3>
:Mosel :hasProduction "4500000000"^^xsd:double.
#@ <id4>
:Mosel :hasGeometry "POINT(28.6,77.2)".
#@ <id4>
:Moselle :pollutedBy :pesticide.
#@ <id5>
<id4> :includes ?pest.
#@ <id6>
?pest :concentration ?c.
#@ <id7>
:Moselle :hasMouth :Rhine.
#@ <id8>
:Moselle :source :Vosges_mountains.
#@ <id8>
:Moselle :hasGeometry "LINESTRING((28.3,77.5),(28.4,77.6))".
\end{verbatim}
\label{lst:yagoSubgraph}
\caption{{Example RDF knowledge graph: Describes pollution in wine growing regions of Europe.}}
\label{fig:RDFKnowledgeGraph}
\end{figure}

\begin{figure*}[htb]
\begin{minipage}[c]{0.5\textwidth}
	\includegraphics[width=\linewidth]{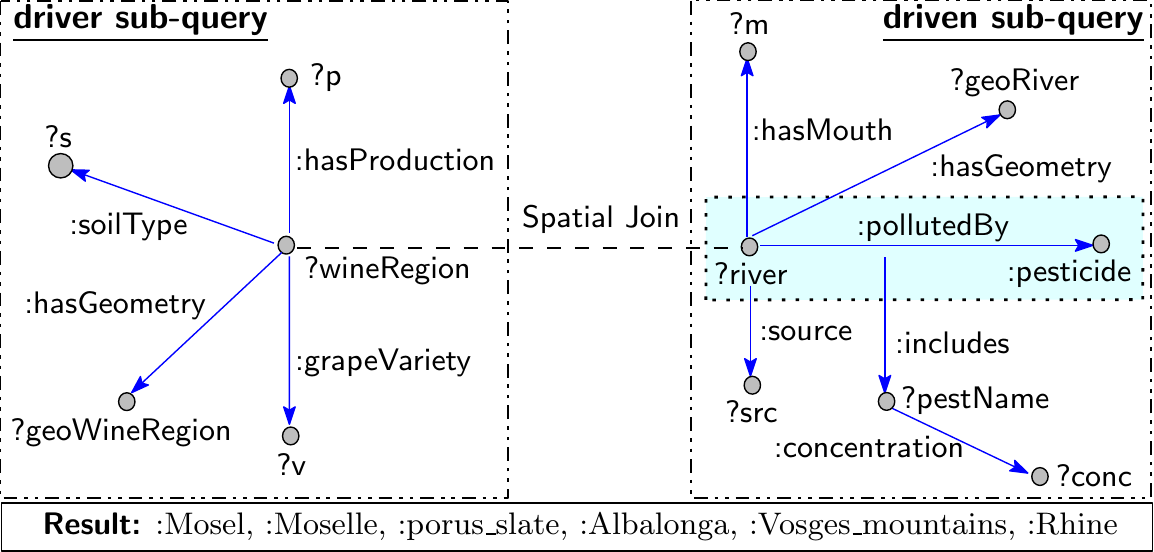}
\end{minipage}\quad\quad
\begin{minipage}[c]{0.45\textwidth}
\centering
\begin{lstlisting}[escapechar=!][frame=single]
SELECT ?wineRegion ?river ?s ?v ?src ?m
  ((f1(?p) * f2(?c)) as ?rank)
WHERE {
  ?wineRegion :grapeVariety ?v.
  ?wineRegion :soilType ?s.
  ?wineRegion :hasProduction ?p.
  !\bfseries{?wineRegion :hasGeometry ?geoWineRegion.}!  
  ?reif rdf:subject ?river.
  ?reif rdf:predicate :pollutedBy.
  ?reif rdf:object :pesticide.
  ?reif :includes ?pestName.
  ?pest :concentration ?c.
  ?river :hasMouth ?m.
  ?river :source ?src.  
  !\bfseries{?river :hasGeometry ?geoRiver.}!  
  FILTER(!\bfseries{distance(?geoWineRegion,?geoRiver)<"10"}!)    
} ORDER BY DESC(?rank) LIMIT 1
\end{lstlisting}%
\end{minipage}
\caption{{Running Example Query: Finds the high-producing wine growing regions of a given soil type and grape variety (driver sub-query) and rivers with high pollution levels and specified mouth and source (driven sub-query), which are no more than 10kms apart (spatial join). The results are ranked based on product of amount of wine production and the river pollution concentration (top-$k$).}}
\label{fig:topksparql}
\end{figure*}

As a running example, we use a reified RDF listing of the knowledge graph shown in Figure~\ref{fig:RDFKnowledgeGraph}. It shows a snippet which contains information about pollution in wine growing regions of Europe. Figure~\ref{fig:topksparql} illustrates a typical top-$k$ SPARQL query on this knowledge graph. This query produces the details of the top-$1$ wine producing region situated near a river, ranked using a function combining its production and pollution levels.

\section{\streak} \label{sec:ourFramework}
In this section, we describe in detail our \streak system for processing top-$k$ spatial join queries. \streak follows exhaustive indexing approach introduced by RDF-3X~\cite{neumann2008rdf}, and extended to support reified RDF stores~\cite{rqrdf3x,quarkX}. This results in indexes over permutations of \textit{subject}, \textit{predicate}, \textit{object} and \textit{reification id} assigned to each statement in the knowledge graph. \streak also maintains an index containing block-level summary of various numerical attributes useful for top-$k$ processing of user-specified scoring function~\cite{quarkX}. We will not discuss the details of these indexes and RDF string / URI encoding strategies, which have been detailed elsewhere~\cite{neumann2008rdf,neumann2009scalable,rqrdf3x,quarkX,encodingSpatialRDF}.

\streak introduces \squad, a novel RDF-specific extension of Quadtree~\cite{finkel1974quad} -- a popular in-memory spatial index structure. \squad stores all the spatial entities in the knowledge graph, and embeds the associated soft-schema information. Its structure along with the associated spatial entity identifier encoding are explained in Section~\ref{subsec:treeBuilding}. 

Next, in Section~\ref{subsec:traverseLeft}, we present the details of the spatial join algorithm which leverages \squad in a careful manner to balance the CPU and IO costs while processing the spatial filter specified in the query. We adopt the terminology from the adaptive query processing literature~\cite{deshpande2007adaptive} and term the left sub-query of the spatial join (see Figure~\ref{fig:topksparql}) as \textbf{driver sub-query}, and the right sub-query as \textbf{driven sub-query}. As part of the spatial adaptive query processing, explained in Section~\ref{subsec:aqp}, \streak chooses the driver and driven sub-query from the original query. 

\subsection{\squad Index for Spatial Entities} 
\label{subsec:treeBuilding}

Underpinning \streak is the \squad index that is designed with compactness and
efficiency in mind, to drive a dynamic query execution strategy that changes
plans during run-time. \streak combines Z-order locality preserving layout for
identifier encoding, quadtrees~\cite{finkel1974quad} for spatial partitioning
and indexing, and characteristics sets for capturing the soft-schema
information related to spatial objects in the corresponding partition. The
intuition behind combining all this information in a single index is motivated
by both the nature of locality within spatial data and the schema-less,
single table philosophy of the RDF. The resulting index, called the \squad
index, enables a more deeply integrated query processing strategy that helps in
early pruning as well. We illustrate the structure of \squad for our running
example query in Figure~\ref{fig:full_quadtree}. 

We note that R-Trees are a popular choice for indexing spatial data. However,
with the design of our indexes and our adaptive query processing strategy, we
observed that \streak outperforms R-Trees in this RDF setting (validated in the
experimental section).

\begin{figure*}[bth]
  \begin{center}
    \includegraphics[width=\linewidth]{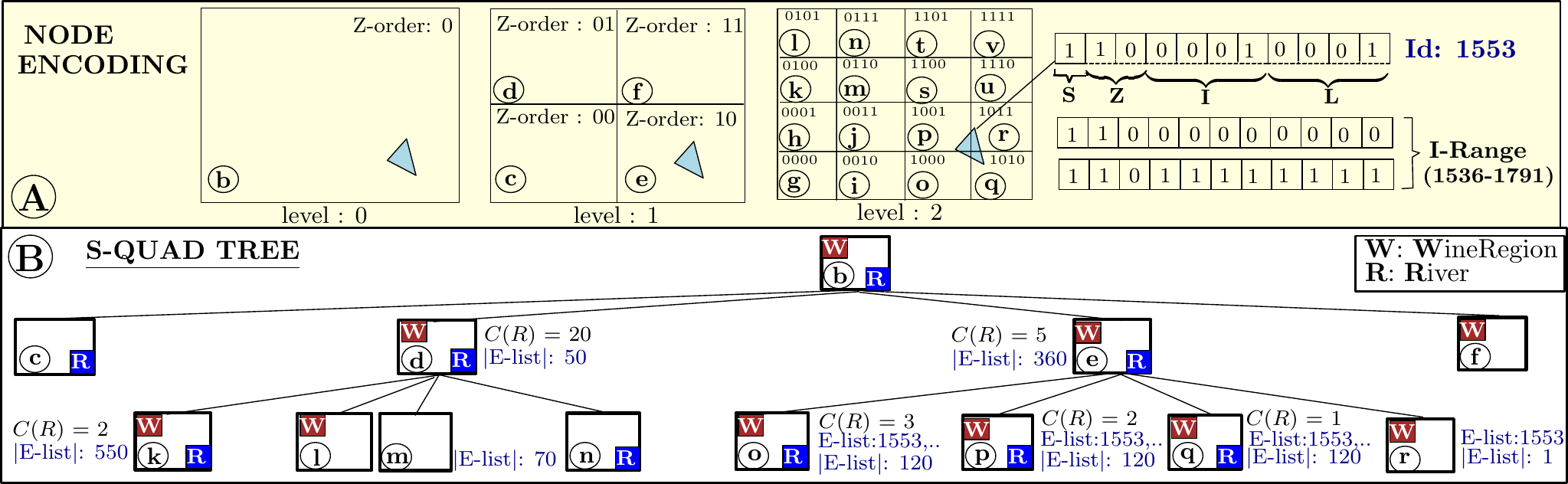}
    \caption[\streak: Toy \squad]{{\textbf{Node encoding:} The triangle is assigned the
identifier based on {\it (z-order = 2, local id = 1, level = 1)}.  \textbf{\squad:}
where each node captures E-List, I-Range, Characteristic Sets for River {\bf
(R)} and Wine {\bf (W)}, Cardinalities for Spatial CS for river {\bf C(R)}, and
MBR (not shown)}}
    \label{fig:full_quadtree}
  \end{center}
\end{figure*}

\subsubsection{Identifier assignment and encoding of spatial entities} Instead of the traditional identifier assignment (e.g., lexicographic~\cite{neumann2008rdf} or compressive~\cite{2016arXiv160404795U}), we assign a unique identifier for each spatial entity that inherently captures its spatial characteristics within the hierarchical partitioning imposed by the quadtree. For both points and polygons,
the identifier value corresponds to the deepest node in the quadtree that fully contains the object. For some polygons, the node that fully contains them 
could be the root of the tree as well. In any case, we encode the identifier
value compactly in the following format: ${\mathbf{(S, Z, I, L)}}$ -- cf. illustrative encoding of the triangular object in Figure~\ref{fig:full_quadtree}(A).

\begin{itemize}
\item $\mathbf{S:}$ is the most-significant bit (MSB) of the identifier. This
indicates whether the identifier corresponds to a spatial or non-spatial
entity. This allows spatial entities --or, more precisely, fact permutations with leading spatial entities-- to be clustered together in storage.

\item $\mathbf{Z:}$ is the Z-order of the deepest-level node that
completely encloses the spatial object. This ensures that in the \textit{ID}
space, those spatial objects with the same \textit{Z-order prefix} are
clustered together for efficient retrieval later.

\item $\mathbf{I:}$ is the local identifier (incrementally assigned) of the
object that distinguishes from other objects within the same node. In the case
of skewed datasets, when the number of spatial objects causes an overflow in the
local identifier bits, then, such identifiers are assigned to the parent node.
The algorithms specified in this paper work even in the above case, without
any loss of generality.

\item $\mathbf{L:}$ is the level / height of the node within the quadtree. The
number of bits occupied by \textit{Z} is $2^{L}$, and the rest of the bits
capture the local identifiers. 
We limit the maximum number of levels in the quadtree to 10 and therefore
${|L|}$ is fixed to 4. We found that there is little benefit in
partitioning a node to have more than a million ($4^{10}$) quadrants. 
\end{itemize}

Our representation not only ensures that the quadtree remains
lightweight, but also enables co-location of objects in space. By design, we
impose equivalent hierarchies for both quadtrees and Z-curve encoding to
capture the identifiers in the above form. For example, the triangle in Figure~\ref{fig:full_quadtree}(A) is completely
enclosed by node $e$ (deepest node), therefore this polygon is assigned the {\it z-order 2}.

\subsubsection{Storing spatial entities in I-Range and E-list} After assigning identifiers, we start constructing \squad over all spatial entities following the standard quadtree construction algorithm~\cite{finkel1974quad}. During this process, we store the indexed identifiers at each node of the \squad in two different ways so as to help in aggressive pruning of results during query evaluation: 
\begin{asparaenum}[(1)]
\item \textit{I-Range} stores the range (min and max) of
object identifiers which completely lie within the node or any of its children.
Note that the \textit{range} of these objects can be identified from the
Z-order of the node.
	\item \textit{E-List} stores explicitly the identifiers of spatial objects overlapping the node but not fully contained
within the node.  
These
objects overlap more than one child node (necessarily) and therefore are
assigned the Z-order of the parent node. Storing them in E-list helps identify the exact child nodes in the \squad with
which these objects overlap, and thus helps save processing time since prefix range comparisons are faster. These lists are typically quite small as we observed during our experiments.
\end{asparaenum}

Referring to the example in Figure~\ref{fig:full_quadtree}, quadtree nodes $o$, $p$, $q$ and $r$ do not completely enclose the triangle, therefore this object ($ID: 1553$) is explicitly stored in
them -- making the indexed object an \textit{E-list} object in these nodes. 
On the other hand, the same object lies completely inside node $e$, and
hence this object lies in the \textit{I-Range} $[1536, 1791]$ of $e$.

\textit{I-Range} and \textit{E-list} offer two complementary ways to prune
intermediate results during processing. The role of \textit{I-Range}
and \textit{E-list}, and the trade-offs for both during adaptive query execution
are clarified in  subsequent sections.

\subsubsection{Use of Characteristic Sets (CS) in \squad} RDF data usually has a \textit{soft-schema},
stemming from the fact that multiple RDF statements are used to describe the
properties of a subject, and these statements often repeat for similar
subjects. To illustrate this, consider the following query patterns for the
query given in Figure~\ref{fig:topksparql} --- \textsf{(?wineRegion
:grapeVariety ?v), (?wineRegion :soilType ?s), (?wineRegion
:production ?p)}. These statements describe a wine producing region's soil,
variety of grapes grown and soil type, through predicates
:grapeVariety, :soilType and :production. These three predicates are also
strongly correlated in the database since this pairing of predicates holds for
only, but many, wine producing regions. Other entities can likewise be uniquely
identified by the predicates connected to them. This observation has been used
earlier for improving the cardinality estimates of RDF
queries~\cite{characteristicSets}, where they name the set of predicates
connected to an entity as its characteristic set. 

In order to leverage the soft-schema for efficient spatial joins, we store in
each node of quadtree, three types of characteristic sets of all objects
enclosed by that node.
{
\begin{asparaitem}

\item \textbf{Self:} The CS that identifies a spatial object.
For example, for the object \textit{:Mosel}, its CS is
\textit{Vineyard}.

\item \textbf{Incoming:} The CS that are incoming towards the
spatial entity. With \textsf{(:Hochmoselbr\"ucke, :isLocatedIn, :Mosel)}, the incoming
entity is \textit{:Hochmoselbr\"ucke}. And the CS that
identifies \textit{:Hochmoselbr\"ucke} is \textit{Bridge}.

\item \textbf{Outgoing:} The CS that are outgoing from the
spatial entity. For \textsf{(:Mosel, :isIn, :Germany)}, the CS corresponding to
\textit{:Germany} is \textit{Country}.

\end{asparaitem}
}

Note that these characteristic sets are stored in Bloom filters~\cite{Broder02networkapplications} for space efficiency.

In Figure~\ref{fig:full_quadtree}(B), nodes \textit{d} and \textit{e} intersects
spatial objects described by both characteristic sets \textbf{W} and
\textbf{R}, and node \textit{f} intersects spatial objects described by
characteristic set \textbf{W} only. Hence, when performing spatial join between
\textit{driver} sub-query (described by \textbf{W}) and \textit{driven}
sub-query (described by \textbf{R}), the spatial join algorithm needs to visit
only the children of the nodes \textit{d} and \textit{e}. Thus, storing the
CS of spatial objects within the nodes of the quadtree helps
perform focused traversal of the quadtree, and reducing comparisons. 

Next, for each characteristic set at each node, we find the objects of CS that  intersect with the node's spatial extent and maintain their aggregate cardinalities in nodes of \squad. As we discuss later, these statistics help in cost-estimation of spatial joins.

Finally in \squad, the minimum bounding rectangle (MBR) of spatial objects associated with a node is stored at each node. This helps eliminate comparisons over regions of quadtree nodes that are not occupied by spatial entities.

We remark here that by their nature quadtrees --and thus \squad-- are relatively easy to update since it affects only the small number of nodes which overlap with the updated object, without any need to balance the tree. 

\subsection{Spatial Join Algorithm in \streak} \label{subsec:traverseLeft}
We now describe the spatial-join algorithm of \streak which employs \squad for  its efficient implementation. The join algorithm follows the block-wise processing wherein bindings for the driver sub-query are retrieved in blocks. For each block, candidate blocks from the driven plan are found using \emph{sideways information passing}, with which the actual spatial-join is performed. In this process, \squad plays a central role for pruning candidate blocks (i.e., sideways information passing), leading to significant reduction in the number of useless joins. In the following, we describe the three phases of our join algorithm, and use the example from Figure~\ref{fig:full_quadtree}, with the driver sub-query as \textbf{W}\textit{ine producing region} and the driven sub-query as \textbf{R}\textit{iver}.

\renewcommand{\L}{\mathcal{L}}
\newcommand{\V}{\mathcal{V}}
\newcommand{\ie}{{\it i.e.}}
\mathchardef\mhyphen="2D

\subsubsection{Phase 1: Obtaining candidate nodes} In this stage block-wise bindings retrieved from the driver sub-query are used to {\em filter out} spatial regions, essentially \squad nodes, that do not contain results of the spatial join. This is done by traversing the \squad from its root to its leaf nodes by following only nodes that satisfy both the bindings from the driver sub-query as well as characteristic sets of the driven sub-query; we denote this set of nodes by $\V$.
For example in Figure~\ref{fig:full_quadtree}(B), at level 1, we observe that nodes $c$ and $f$ do not satisfy the join as either \textbf{W} or \textbf{R} is missing. However, nodes  $d$ and $e$  participate in the join between \textbf{W} and \textbf{R}, thus \streak traverses only these nodes.

We use $\V_l$ to denote the nodes in $\V$ in level $l$. Due to the nature of the explicitly encoded objects associated with a node, all the spatial objects that are results of the join are contained in $\V_l$ for any level $l$, and in particular, in the leaf-nodes of $\V$ (which we denote by $\V_L$).

\subsubsection{Phase 2: Filtering driven sub-query.}
The \textit{I-Range} and \textit{E-list} associated with a node in the \squad, say $a$, store the spatial objects associated with $a$ and can be used for filtering the driven sub-query using {\em sideways information processing}. However, there may be local density variations in different nodes of the \squad. Therefore, it is important to carefully choose the nodes in $\V$ to be used for filtering. Going forward, we first discuss the process of filtering and associated CPU and IO costs. Then we describe our algorithm to choose appropriate nodes in $\V$ that can be used for efficient filtering.

Next we explain how \textit{I-Range} and \textit{E-list} help in filtering the driven sub-query. \textit{I-Range} and \textit{E-list} helps skip entries which are retrieved from index scans of driven sub-query. Recall from Figure~\ref{fig:full_quadtree}(A) that \textit{I-Range} of a node is the set of identifiers which completely lie within that node e.g. \textit{I-Range} of node \textit{e} is from \textsf{1536} to \textsf{1791}. An RDF fact containing a spatial entity is skipped if it does not lie within \textit{I-range} or \textit{E-list}. Thus \textit{I-Range} and \textit{E-list} can be used for skipping irrelevant entries in the indexes. Note that skipping entries using \textit{I-range} is faster --- to understand this consider the case when the current block does not contain values in between \textit{I-range}, then a skip can be made to the next block by skipping all irrelevant entries in between. However, skipping using a large \textit{E-list} can be expensive, as each RDF fact needs to be potentially checked for skip. We observed that this causes high CPU overhead, as it disrupts the high cache locality of sequential scans~\cite{neumann2009scalable}. Therefore, we associate a CPU cost with $a$ that is proportional to the cardinality of $E\mhyphen list(a)$.

Take for example, Figure~\ref{fig:full_quadtree} in which we assume that $\V$ is the set of all nodes containing both \textbf{W} and \textbf{R}. We observe that the size of \textit{E-list} at child node $k$ is greater than at node $d$. However, notice that node $d$ satisfies the join with estimated cardinality of \textit{C(R)} being 20 whereas, its children node $k$ has estimated cardinality of \textit{C(R)} as 2. Therefore, with respect to fetching of blocks, it is better to filter using the \textit{I-range} and \textit{E-list} of $k$ compared to those of $d$. We capture this by defining the IO cost of a node $a$ as the cardinality of the characteristic set stored at $a$ and represent it as $|CS(a)|$.

Observe that the spatial objects associated with $a$ are also associated with the parent of $a$ -- so it may not be prudent to filter using both $a$ and its parent. Furthermore, the IO cost of a node decreases as the level increases whereas CPU cost increases. Therefore, it is necessary to compute an optimal set of nodes for filtering, say $\V^* \subseteq \V$ such that (a) the nodes in $\V^*$ collectively cover all spatial objects associated with nodes in $\V$, and (b) have minimal IO and CPU footprints. We should mention that nodes in $\V^*$ may come from different levels, e.g., a few possibilities of $\V^*$ in Figure~\ref{fig:full_quadtree}(B) are $\{$\textsf{b}$\}$, $\{\mbox{\textsf{d, e}}\}$, $\{$\textsf{d, o, p, q}$\}$, $\{$\textsf{e, k}$\}$.
To compute $\V^*$, first we associate a $cost(a)$ with every node $a \in \V$
based on the above discussion.
\begin{center}
$$cost(a) = \alpha_{\scriptscriptstyle IO} \overbrace{|CS(a)|}^{\text{IO cost}} + \alpha_{\scriptscriptstyle CPU} \overbrace{|E\mhyphen list(a)|}^{\text{filtering}}$$
\end{center}

Next, for any node $a$, we define $\V^*(a)$ as the optimal set of nodes from $\V$ in the sub-tree rooted at $a$. For computing the cost of $\V^*$, we also need to take into consideration the inherent (CPU) cost of combining the large \textit{E-list}s of nodes in $\V^*$. Obviously, this cost is 0 if there is none or only one {\it E-list} in consideration and is proportional to the size of the combined list (we keep the lists sorted) if there are two or more lists; we use $\alpha_{\scriptscriptstyle merge}$ as the proportionality constant and denote $\alpha_{\scriptscriptstyle merge}|E\mhyphen list(a)|$ by $\xi(a)$. Algorithm~\ref{pseudoDepthEst} explains how to compute the optimal set of nodes by recursively computing the values $\V^*(a)$ for every node $a$, starting from the leaf nodes, based on the following theorem. The theorem uses two additional recursive functions $\sigma^*(a)$ and $\xi^*(a)$. The former represents the cost of $\V^*(a)$ and the latter represents $\sum_{j \in \V^*(a)} \xi(j)$.
In the following theorem, $r$ denotes the root of the \squad, $\gamma(a)$ denotes the children of $a$ that also belong to $\V^*$ and 
$\mu(a)$ stands for $\sum_{j\in \gamma(a)} \xi^*(j)$ if $|\gamma(a)|>1$, and is $0$ otherwise.

\begin{theorem}\label{thm:optimal-subset}
    $\V^*=\V^*(r)$ can be recursively computed using recurrences~{\ref{cutSelectionEq1} and \ref{cutSelectionEq2}}. Furthermore, complexity of computing this is linear in the number of nodes of the \squad.

\noindent For leaf node $a$,
\begin{align}\label{cutSelectionEq1}
\V^*(a), \sigma^*(a), \xi^*(a) = \left\{%
    \begin{array}{@{} ll @{}}%
	    \{\}, 0, 0 &: \text{$a \not \in \V$}\\
	    \{a\}, cost(a),\xi(a) &: \text{$a \in \V$}
	\end{array}
    \right.
\end{align}
\noindent For non-leaf node $a$,
\begin{align}\label{cutSelectionEq2}
    \parbox[c]{.045\textwidth}{$\V^*(a)$,\\$\sigma^*(a)$,\\$\xi^*(a)$} = \left\{
	\begin{array}{l}
	    \{\},0,0\qquad: \text{if } a \not\in \V\\
	    \\
	    \displaystyle \left.\begin{array}{@{} l @{}}
		\{a\}\\
		cost(a)\\
		\xi(a)
	\end{array}\right\}\text{if }\displaystyle cost(a) < \mu(a) + \sum_{j \in \gamma(a)} \sigma^*(j)\\
	    \\
	    \left.\begin{array}{@{} l @{}}
		\bigcup_{j \in \gamma(a)} \V^*(j)\\
		\sum_{j \in \gamma(a)} \sigma^*(j) + \mu(a) \\
		\sum_{j \in \gamma(a)} \xi^*(j)
	    \end{array}\right\} \text{otherwise}
	\end{array}
    \right.
\end{align}
\end{theorem}

\begin{proof}
The cases for nodes not in $\V$ are obvious and included for completeness.
We will discuss the other cases with the help of Figure~\ref{fig:possibleCuts_Child_InnerNode}.
When $a$ is a leaf-node in $\V$, e.g., the left child of $E$, the optimal $\V^*(a)$ must be $\{a\}$ since it is the only option. $\sigma^*(a)$ and $\xi^*(a)$ are accordingly assigned.

\begin{figure}[bth]
\begin{center}
\includegraphics[width=0.6\linewidth]{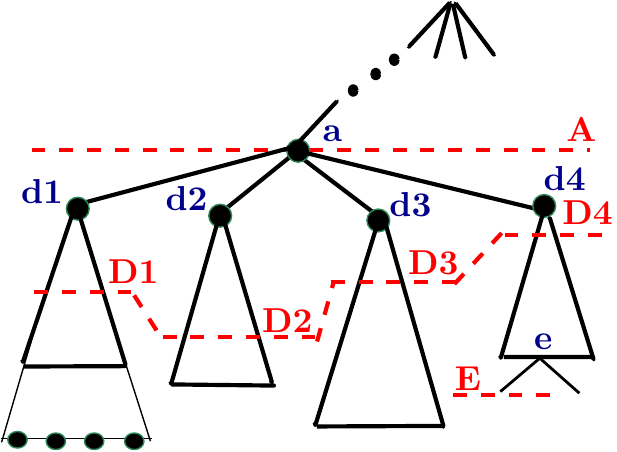}
\caption[Selecting optimal nodes in \squad]{Selecting optimal $\V^*$ in \squad
    \label{fig:possibleCuts_Child_InnerNode}}
\end{center}
\end{figure}

For the other cases, consider any internal node in $\V$, e.g., $a$ in Figure~\ref{fig:possibleCuts_Child_InnerNode}, with four children. As the Figure shows, there are two possibilities for choosing $\V^*(a)$. Either it contains $a$ (second case of recurrence~\ref{cutSelectionEq2}), in which case it should not contain any other node since every spatial object associated with any node in the subtree of $a$ is already associated with $a$. Otherwise, if $a \not\in \V^*(a)$ (third case of recurrence~\ref{cutSelectionEq2}), then we should select the best possible sets in each of $d_i$. The latter are nothing but $D_i$ (note that some of them could be empty) since we have to cover all spatial objects associated with nodes in $\V_L$ and, furthermore, any set of nodes in $d_i$ other than $D_i$ has a cost larger than that of $D_i$. The cost of selecting $a$ in the first case would be simply $cost(a)$. The cost in the second case involves (i)~the cost for selecting the $D_i$s as well as (ii)~that of merging their $E\mhyphen list$s. The cost for case~(ii) is 0 if there is at most one $E\mhyphen list$ and otherwise, is proportional to the total number of objects in all those $E\mhyphen list$s. $\mu(a)$ computes this exact quantity; notice that $\mu(a)=0$ if all but one $D_i$ are empty and one $D_i$ has only one node.

Since $|\gamma(a)| \le 4$, $\mu(a)$ can be computed in constant time and so can be $\sigma^*(a)$, $\xi^*(a)$, $\V^*(a)$ (assuming an efficient implementation of union operation for sets). Therefore, $\V^*=\V^*(r)$ can be computed in linear time in the number of nodes in the \squad.
\end{proof}

\begin{algorithm}[!hbt]
\caption{Computing optimal nodes for filtering}
\label{pseudoDepthEst}
\begin{algorithmic}[1]
\small 
\State \textbf{Input}: candidate nodes $\V$, parameters $\alpha_{\scriptscriptstyle IO}$, $\alpha_{\scriptscriptstyle CPU}$, $\alpha_{\scriptscriptstyle merge}$
\State \textbf{Output}: optimal set of nodes $\V^*$
\For {each node $a$, starting from leaves, in a bottom-up fashion}
\If{$a$ is a leaf node}
    \State Compute $\V^*(a)$, $\sigma^*(a)$ and $\xi^*(a)$ using eq. \ref{cutSelectionEq1} 
\Else
    \State Compute $\V^*(a)$, $\sigma^*(a)$ and $\xi^*(a)$ using eq. \ref{cutSelectionEq2} 
\EndIf
\EndFor
\State \Return $\V^*(r)$
\end{algorithmic}
\end{algorithm}

Algorithm~\ref{pseudoDepthEst} gives pseudo-code of our dynamic programming algorithm that essentially uses the recurrences in Theorem~\ref{thm:optimal-subset} to compute $\V^*$ in an efficient manner.
These optimal sets of nodes are selected for each block of values retrieved from the \textit{driver} sub-query to be subsequently used for filtering.

\subsubsection{Phase 3: Performing the Spatial Join}
Next in order to perform spatial join \streak retrieves spatial objects from the driven sub-query. Naturally \textit{driven} sub-query needs an optimal plan for efficient execution. It determines the optimal plan adaptively using a technique coined \textit{APS} (explained next in Section~\ref{subsec:aqp}). It then performs join between spatial objects retrieved from the driver and driven sub-queries using \squad starting at $\V^*$ as described below.  At each level, objects retrieved from the driver and driven sub-queries --- coined \textit{driver object} and \textit{driven object} --- can overlap with the MBR of none, one or several nodes:
\begin{itemize}
\item \textbf{Overlap with no MBR:} If (i) driver object and driven objects do not overlap with a common MBR, and (ii) driven object is not contained in the node overlapped by driver object, then safely filter both driver and driven objects. This is because driver and driven objects do not intersect with each other and hence can be excluded from further consideration.
\item \textbf{Overlap with one or more MBRs:} If (i) driver object and driven object overlaps with a common MBR, and (ii) driven object is contained in the node overlapped by driver object, then the algorithm recursively checks the children of the node till the level where the diagonal length of the node's MBR is equal to the query distance.
\end{itemize}
In brief, the algorithms keeps finding whether the driver and driven objects intersect or not till it reaches a level -- where the diagonal length of the node's MBR is equal to the query distance. 

\subsubsection{Refinement Step} 
The spatial join gives pairs of candidate spatial objects, however recall that \streak approximated the spatial objects by minimum bounding rectangles (MBRs), hence the distances measured need to be validated. Moving forward we pass these candidate spatial objects to \textit{refinement step} which is performed at the end of every block-wise query execution cycle. During \textit{refinement step}, \streak validates the distance join constraint using object's exact representation. 

\subsection{Adaptive Query Processing for Top-$K$ Spatial Joins}  \label{subsec:aqp}

We begin by first discussing the data access choices available to \streak. A naive way would be to retrieve all tuples from the driver and driven sub-queries and join them. This would naturally will be inefficient for complex queries on large datasets. Another alternative could be to join the driver and driven sub-queries in a tuple-at-a-time manner. This would again be inefficient for complex RDF queries containing many joins because of (1) increased number of function calls required (2) pruning performed by each join operator. The third alternative is to use popular block-wise query execution of columnar stores, which amortizes access costs over a number of tuples. \streak adopts this alternative because of its proven advantages in columnar stores.

The complete block-wise query execution pipeline of \streak is shown in Figure~\ref{fig:queryProcessing}. \streak begins by accessing driver sub-query in block-wise manner, by retrieving block $blk$. It then joins this block with the driven sub-query to find the complete score of tuples based on the given ranking function. These scores are used to compute the threshold $\theta$, which is the score of $k$th result tuple, and the algorithm stops after $k$ results have been found with score above $\theta$ (for descending). Observe that the driven sub-query can joined either block-wise (by retrieving block $blk'$) or using full-scans. These plans, termed as \textit{N-Plan} and \textit{S-Plan} respectively, are adaptively chosen at run-time based on cost-estimates (described later).

Block-wise query processing brings in two more advantages during this phase. First, operating at block-wise granularity offers a balance between low-overhead of selecting different plans for different portions of input relations (\textit{blocks}), according to ranking functions, and reasonable accuracy of cost estimation. Second, the block-wise bounds of numerical attributes and characteristic set information stored in the nodes of the \squad enable us to perform early termination and avoiding self-joins, respectively. The primary function of our adaptive query plan generation algorithm, called \textbf{Adaptive Processing for Spatial filters or APS} (pronounced ``apps''), is thus to identify and route blocks of tuples through customized plans based on statistical properties, relevant to the query execution strategies. 

The traditional approach to query optimization involves choosing an optimal
plan based on statistics prior to query execution, and using that selected plan at runtime. This, however, would end up processing and producing (1) far more intermediate
results than is necessary (2) the entire input cardinality -- which is expensive. \jyoticomment{Even as plan generation is a standard problem in database systems, the schema-free nature of RDF data, along with spatial and top-$k$ join operations, complicate generation and selection of an optimal plan even
further.} Previous top-$k$ query optimization work~\cite{ilyas2004rank} attempted to solve top-$k$ plan generation problem based
on a strong assumption that each tuple in driver sub-query is equally
likely to join with tuples in driven sub-query. However, this does not
frequently hold in tightly coupled triple patterns found in SPARQL queries. Instead, we argue that an
adaptive query processing strategy (AQP) can be a better strategy for these
workloads, as for each block retrieved it will help terminate early and determine the plan based on the selectivity of the block.

\subsubsection{Challenges} Between an optimal plan selection of top-$k$ spatial join
queries with AQP and its corresponding efficient execution, lies several
challenges that need to be overcome.

\begin{enumerate}
  \item In the case of top-$k$ queries, where spatial join essentially
        comprises a self-join between two parts of a dataset, cost estimation 
        is difficult.
  \item Different regions of the data have distinct statistical 
        characteristics and a single plan for all regions performs poorly.
  \item Early out feature of top-$k$ queries requires estimation of input cardinalities of these index scans thus making it even harder to choose an
        optimal plan, and in general, poses many challenges in costing
        rank-join operators.
  \item Given the verbose nature of RDF, finding optimal join order of SPARQL queries is challenging because the number of possible query plans is in the order of factorial of number of index scans required. With many real-world queries containing a hundred index scan, it is very costly to find an optimal plan~\cite{sahoo2010semantic}.
\end{enumerate}

We will see in this section how the design of our indexes naturally enable us
to devise a more effective, adaptive execution strategy for query processing.  

\subsubsection{Plan Selection:} 

We address the above mentioned challenges in \streak, by first breaking down
the plan selection problem into that of finding the least-cost for the
\textbf{driver plan} and the \textbf{driven plan}. We use heuristics and a
cost-based optimization approach~\cite{selinger1979access} to select the query
plan for the \textit{driver sub-query}. The heuristic pushes numerical
predicates deep in the query plan and reduces the search space of join order
permutations. This helps in early-termination, and consequently helps reduce
the query execution time in majority of the cases for top-$k$ queries. This
heuristic has also been used in earlier works SPARQL-RANK~\cite{sparqlrank} and
Quark-X~\cite{quarkX}. For the remaining driver sub-queries, a cost-based
join-order optimization is used.

\begin{figure}[bth]
\begin{center}
	\includegraphics[width=\linewidth]{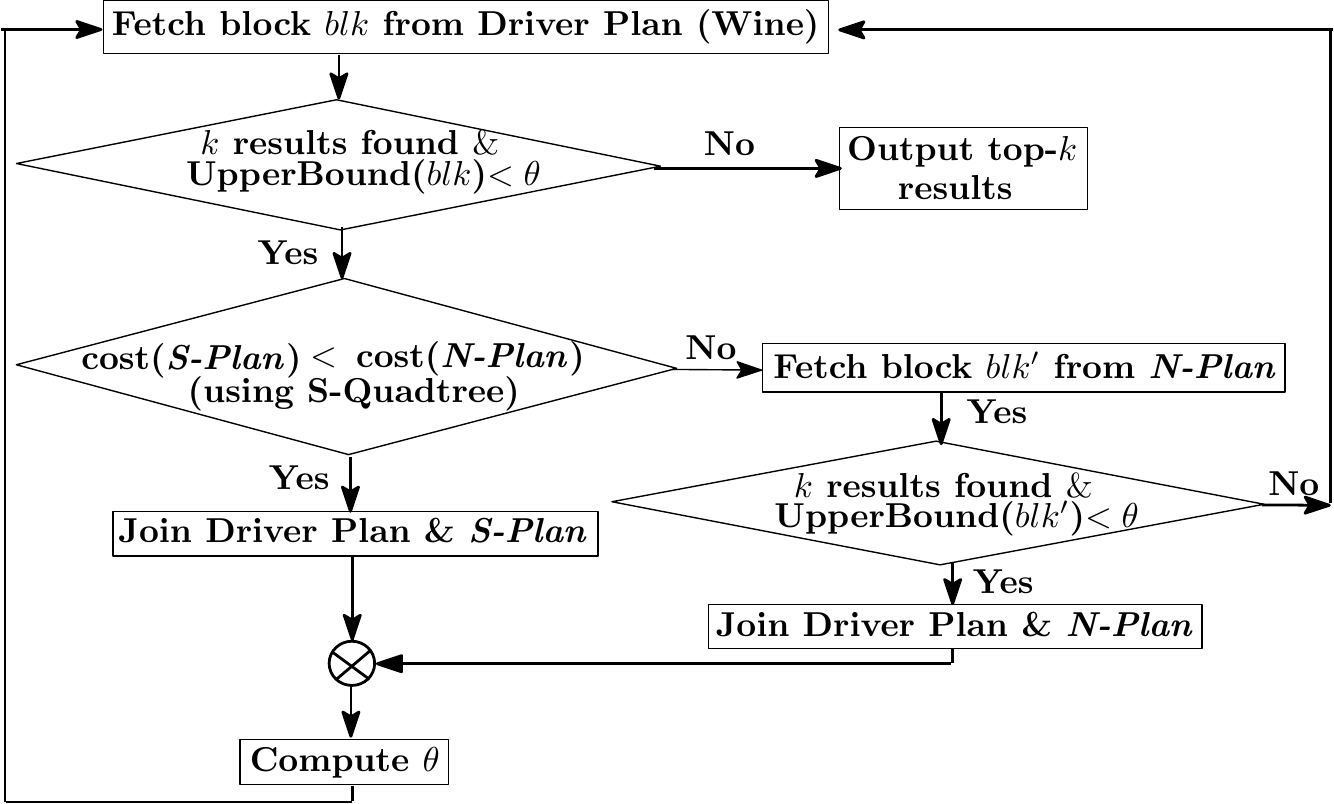}
\caption[\streak: Query Processing Flow-Chart]{{Query Processing Flow-Chart}}
\label{fig:queryProcessing}
\end{center}
\end{figure}

The need for adaptively choosing customized plans for driven sub-query can be
well motivated using our running example, which retrieves the top-most wine
producing regions that are located within $d\ kms$ of a river. This involves a
spatial join between regions containing vineyards (\textit{driver} sub-query)
and rivers (\textit{driven} sub-query). In our example, the wine growing
regions are the Gobi desert in China, which has no river flowing across it, and
Baden in Germany that has its vineyards directly overlooking the Rhine river.
An efficient implementation of top-$k$ spatial distance join operator should
therefore perform spatial join with Gobi desert (in \textit{driver} sub-query) first, as the spatial join will help in eliminating Gobi desert quickly. We call the driven plan generated using
this heuristic of pushing spatial joins deep in query execution, as
\textbf{S-Plan}.  This is shown in Figure~\ref{fig:plans}(B), with driver and
driver sub-query shown in Figure~\ref{fig:plans}(C) and (D), respectively. This
heuristic significantly reduces query execution time when the spatial join
operator is highly selective.

Secondly, the spatial-join operator should check for early-termination first 
when blocks -- from \textit{driver} sub-query -- containing spatial regions like Baden are retrieved, as Baden is
not the top-most wine growing region. Hence pushing numerical predicates deep
in the \textit{driven} sub-query's plan will help in eliminating such tuples. We call such a driven
plan generated using this heuristic as an \textbf{N-Plan}, as shown in
Figure~\ref{fig:plans}(A), with its driver and driven plans shown in
Figure~\ref{fig:plans}(C) and (E) respectively. 

For top-$k$ spatial join workloads both scenarios are highly common. Starting
with either heuristic as seed, larger plans are created by joining optimal
solutions of smaller problems that are adjacent in driven sub-query.

\begin{figure}[bth]
\begin{center}
	\includegraphics[width=\linewidth]{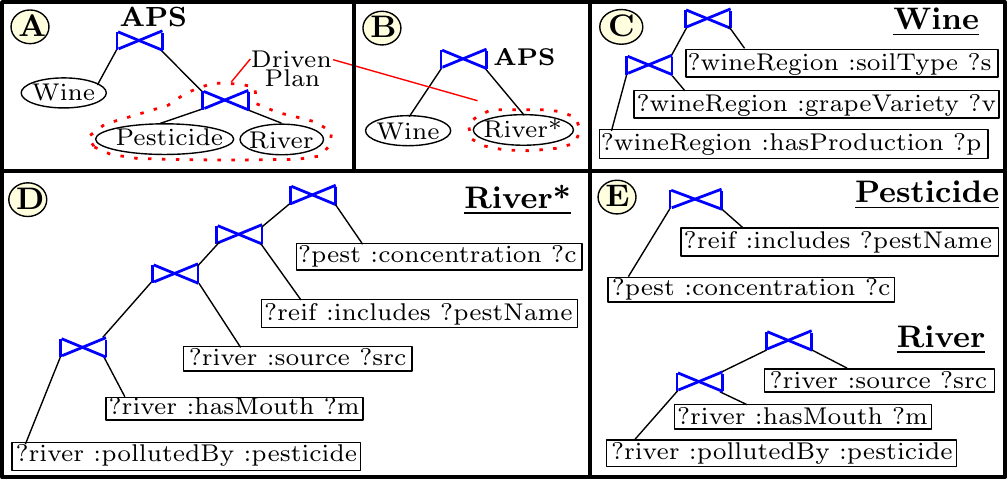}
  \caption[\streak: Possible Plan Choices]{{Possible Plan Choices: (A) N-Plan, with rivers filtered
    first by numerical predicate (pesticide concentration), (B) S-Plan, with a spatial join
    between vineyards and rivers done first, (C) driver plan, (D) Plans supporting
    S-Plan heuristics, and (E) Plans supporting N-Plan heuristics}}
\label{fig:plans}
\end{center}
\end{figure}

\streak then computes the cost of the two plans and routes the tuples through
least cost plan. \streak's block-wise query execution enables it to switch
among the two plans at runtime -- with zero cost of switching plans at
materialization points --- points where next blocks are retrieved and executed
in driver sub-query. At these points \streak performs the spatial join between
the newly retrieved blocks from the driver and driven plan respectively, which
means that the work done by earlier blocks does not go wasted. This helps
\streak overcome the drawbacks of earlier top-$k$ adaptive query processing
approaches that wasted substantial amounts of work at materialization points
\cite{ilyas2004rank} and used limited AQP by switching only indexes at
run-time~\cite{quarkX}. In contrast, \streak switches entire customized plans at
runtime.

\subsubsection{Cost Model} 
We now discuss \streak's cost-model for choosing amongst driven plans.  The
cost function of our spatial join operator takes into account the cardinality
of the spatial join. Greater the cardinality of the spatial join, the greater
we expect the number of candidate results to be and thus higher the query
execution time. We calculate the cardinality of the join to be a product of the
(1) estimated cardinality of \textit{driven} plan $C(R')$ (discussed below) and 
(2) the number of tuples retrieved from the driver plan (as discussed in
sub-section~\ref{subsec:traverseLeft}). 

\streak uses a simple cost model to estimate the cost and cardinalities($C(R')$) of its customized \textit{driven} plans: N-Plan and S-Plan. Recall that the N-Plan pushes numerical predicates deep in the
query plan and uses block-wise evaluation to achieve early-termination, while
the S-Plan pushes spatial join evaluation deep in the query plan and uses
full-scans for evaluating the join. For the rest of the section, we denote the
time of execution of driven plan $R$ as $T(R)$. Let $R$ and $R_{i}$ denote full and block-wise executions of S-Plan
and N-Plan, respectively. $R_{i}$ corresponds to $R$ when either (a) all blocks
are retrieved, or (b) early-termination is achieved. Putting $R$ and $R_{i}$
together, we have the following equation where $x$ is the estimated number of
blocks required to be retrieved before early-terminating is achieved. 

\noindent To estimate time cost,
\begin{align}\label{cost_calculation}
T(R) = \left\{%
    \begin{array}{@{} ll @{}}%
	    x.T(R_{i}) &: \text{$x < R/R_{i}$}\\
	    T(R) &: \text{$x \geq R/R_{i}$}
	\end{array}
    \right.
\end{align}

We estimate $R_{i}$'s result cardinality as $C(R_{i}) = x . C(R)/nb$, where '$nb$' is the number of blocks of numerical predicate, and $C(R)$ can be found from spatial characteristic sets stored in the nodes of the \squad. The intuition of the equation is as follows: when number of blocks is 1, then $C(R_{i})$ is $C(R)$ i.e. it is equivalent to all intermediate results without any early pruning. We assume that all buckets contribute to $C(R)$'s results equally. So when only a portion of them are selected, $C(R_{i})$ is proportional to the fraction of blocks that are selected for early termination, which is characterized by $x/nb$. To estimate `$x$', we find all blocks whose upper bound is greater than the threshold. $C(R')$ is equal to  $C(R_{i})$ for N-Plan and $C(R)$ for S-Plan.

We next route the tuples through the plan with the least cost, and thus at
runtime seamlessly keep alternating between the two customized plans for each
new block of tuples retrieved from the driver sub-query. It may seem at first
glance that the cost of \textit{N-Plan} will be lower than \textit{S-Plan} as
the cardinality of \textit{N-Plan} is always less than equal to the cardinality
of \textit{S-Plan}. However, we found this to not be the case because repeated
scans, owing to block-wise query processing, increases the cost of
\textit{N-Plan}. We later show experimentally in
Section~\ref{sec:experimentalResults}, that our proposed \textit{APS} algorithm
does in fact help select the least cost customized plan, and accelerates query
processing times by one to two order(s) in magnitude.

\section{Evaluation Framework} \label{sec:eval}
\streak is implemented in C++, compiled with g++-4.8 with -O3 optimization flag. 
All experiments were conducted on a Dell R620 server 
with an Intel Xeon E5-2640 processor running at 2.5GHz, with 64GB main-memory, 
and a RAID-5 hard-disk with an effective size of 3TB. For all the experiments, 
the OS can use the remaining memory to buffer disk pages. In our experimental 
evaluation, we report cold-cache timings after clearing all filesystem caches 
and by repeatedly running the query processor with the same query 5 times, and 
taking the average of the final 3 runs. For warm cache numbers, we repeat the 
same procedure but without dropping caches.

We perform a direct comparison of \streak with other full-fledged systems by 
ensuring the use of spatial indexes for the workloads used in the paper. In 
Virtuoso, we employ a two dimensional R-tree implementation for indexing the 
spatial data and on postgres we have built the ``gist'' (generalized search 
tree) index for supporting spatial data. Postgres supports varied geometry types 
like POINTS, POLYGONS and LINESTRINGS, while Virtuoso only supports POINT 
geometries~\cite{encodingSpatialRDF}. The number of results returned by Postgres 
and Virtuoso is constrained by using LIMIT `k' in the query. 

\subsection{Datasets}
We used two widely used, real-world datasets: YAGO3 Core~\cite{cidr_yago3} and 
Linked Geo Data(LGD)~\cite{auer2009linkedgeodata} for evaluating the performance 
of \streak in comparison with other frameworks. LGD contains freely available 
collaboratively collected data from Open Street Map project. YAGO3 contains 
facts extracted from Wikipedia, GeoNames and WordNet. YAGO3 attaches annotations such as confidence score, time, spatial location, provenance, etc. with the fact identifier of each fact using Turtle format. Table~\ref{tab:data} shows 
the characteristics of the two datasets. We observe that YAGO3 and LGD datasets contain 85 million and 324 million 
quads respectively. The size of the quad-tree used for storing YAGO3 and LGD 
datasets, is 37 MB and 119 MB respectively, about 0.04\% and 2\% of size of raw data. This is despite the fact that 50 percent of facts in LGD contain description about spatial objects, consisting of rich suite of geometries such as POINTS, POLYGONS and LINESTRINGS. To simulate situation where confidence score is attached to each fact, we assign confidence values between 0 and 1.0 using exponential and uniform distributions to all facts in the original datasets. We find that results with uniform distribution follow a similar trend, so for sake of brevity we report the results only for exponential distribution. 

\begin{table}
               \caption{Dataset Characteristics}
               \label{tab:data}
               \centering
               \begin{tabular}{l l l l l }
               \toprule
               Dataset & quads & points & lineStrings & polygons\\
               \midrule
               YAGO3 & $85,928,328$ & $590,000$ & 0 & 0\\
               LGD & $30,857,332$ & $4,000,000$ & $2,600,000$ & $264,000$\\
               \bottomrule
               \end{tabular}
\end{table}

\begin{table*}
\caption{Benchmark Query Characteristics (Join Type: (1) SS: Subject-Subject, (2) RS: Reification-Subject, (3) OS: Object-Subject). }
\label{fig:querychar}
\centering
\begin{subtable}{1\textwidth}
\scalebox{1.0}{
	\begin{tabular}{p{0.8cm} p{1cm} c c p{1.2cm} p{1cm} p{1 cm} p{1.3 cm} p{1cm} p{1.5 cm} p{1.3cm}}
	\toprule
	\textbf{Query id} & \textbf{Shape} & \textbf{Type}	 & \textbf{\# TP} & \textbf{\# Quant
 TP} & \textbf{\# Non Quant TP} & \textbf{\# Joins} & \textbf{Join Type} & \textbf{Join degree} & \textbf{Result Card.} & \textbf{Selectivity}
\\
	\midrule
	1 & Complex & Point Polygon &  6 & 2 & 4 & 4 & (SS, RS) & (2,2,2,2) & $2072177$ & 1.5E-29\\		
	2 & Complex & Point Point  &  6 & 2 & 4 & 4 & (SS, RS) & (2,2,2,2) & $131525$ & 1.1E-29\\		
	3 & Complex & Point Point  &  7 & 2 & 5 & 6 & (SS, RS) & (3,3,2,2) & $66477$ & 5.1E-42\\		
	4 & Complex & Point Point  &  9 & 2 & 7 & 7 & (SS, RS) & (4,3,2,2) & $613142$ & 1.9E-47\\		
	5 & Complex & Point Polygon  &  9 & 2 & 7 & 7 & (SS, RS) & (4,3,2,2) & $362179$ & 2.1E-48\\		
	6 & Complex & Point LineString  &  6 & 2 & 4 & 4 & (SS, RS) & (2,2,2,2) & $830000000$\footnote{Note that, for long running queries (running for greater than 5 hours) we have high water marked the systems with result cardinality of 830000000 and stopped result generation after these many were generated.\label{note1}} & 8.57E-30\\		
	7 & Complex & LineString Point  &  6 & 2 & 4 & 4 & (SS, RS) & (2,2,2,2) & $830000000$\footref{note1} & 3.2E-29\\		
	8 & Complex & Polygon LineString  &  7 & 2 & 5 & 5 & (SS, RS) & (3,2,2,2) & $830000000$\footref{note1} & 2.4E-36\\		
	\bottomrule
	\end{tabular}
	}
	\caption{LGD}\label{tab:sub_first}
	\end{subtable}	

\begin{subtable}{1\textwidth}
	\scalebox{1.0}{	
	\begin{tabular}{p{0.7cm} c c c p{1.2cm} p{1.0cm} p{1.2 cm} p{1.5cm} p{1.4 cm} p{1.8 cm} p{1.2cm}}
	\toprule
	\textbf{Query id} & \textbf{Shape} & \textbf{Type}	& \textbf{\# TP} & \textbf{\# Quant
 TP} & \textbf{\# Non Quant TP} & \textbf{\# Joins} & \textbf{Join Type} & \textbf{Join degree} & \textbf{Result Card.} & \textbf{Selectivity}
\\
	\midrule
	1 & Star & Point Point & 6 & 2 & 4 & 6 & (SS) & (3,3) & $8836947$ & 8.1E-28\\		
	2 & Star & Point Point &  8 & 3 & 5 & 7 & (SS) & (4,3) & $6247$ & 3.2E-33\\	
	3 & Star & Point Point  & 7 & 2 & 5 & 7 & (SS) & (4,3) & $557$ & 2.9E-34\\		
	4 & Star & Point Point &  8 & 3 & 5 & 8 & (SS) & (5,3) & $5604$ & 5.1E-36\\		
	5 & Complex & Point Point &  8 & 2 & 6 & 6 & (OS, RS) & (2,2,2,2,2,2) & $829557455$ & 1.4E-41\\		
	6 & Complex & Point Point &  7 & 2 & 5 & 6 & (OS, SS, RS) & (3,3,2) & $172238$ & 2.8E-33\\		
	7 & Complex & Point Point &  6 & 2 & 4 & 6 & (SS, RS) & (3,2,2) & $13401$ & 7.9E-30\\		
	8 & Complex & Point Point &  7 & 3 & 4 & 5 & (OS, RS, SS) & (2,2,2,2,2) & $1135925$ & 8.4E-38\\		
	\bottomrule
	\end{tabular}
	}
	\caption{YAGO3}\label{tab:sub_first}
\end{subtable}	
\end{table*}
\subsection{Benchmark Query Workloads}
Numerous performance benchmarks have been created in the recent years for 
processing RDF / SPARQL queries. However, none of them cater to queries 
involving reification (or named graphs), spatial and top-$k$ evaluation. 
Therefore, we design a set of benchmark queries containing query features that 
have been found important in literature ~\cite{diversifiedStressTesting, 
encodingSpatialRDF,quarkX}, namely --- \textbf{structural} features, which 
include:
\begin{itemize}
\item \textbf{Shape: } This feature captures the two ways in which triples patterns 	in a SPARQL query can be combined namely, \textit{star} and \textit{complex}. 
Queries that belong to the \textit{star} class have -- for each of \textit{driver} and \textit{driven} sub-queries -- several triple patterns connected with a common subject variable. Indeed in \textit{star} class, queries in the query graphs  -- of each of \textit{driver} and \textit{driven} sub-queries -- are cliques. We classify the remaining queries, which do not belong to the \textit{star} class as \textit{complex} shaped queries. 

Though this is a high level structural classification, we found it to be sufficiently useful to test the robustness of an RDF system. This is because although RDF engines can be accelerated for \textit{star} class queries with efficient merge joins, the same is not necessarily true for \textit{complex} class queries.

\item \textbf{Type: } This feature highlights the efficiency of RDF systems to handle spatial objects with different spatial extents e.g. polygon, line string, point, etc.

\item \textbf{Triple pattern count(\# TP): } denotes to number of triple patterns 		in a query. This feature tests how an RDF system scales with increasing number 		of triple patterns.

\item \textbf{Quantifiable triple pattern count (\# Quant TP): } denotes number of 		triple patterns in a query whose objects are arguments to the ranking function. 		This feature tests how an RDF system scales with increasing complexity of -- more precisely, the number of attributes in -- the ranking function. 

\item \textbf{Non-quantifiable triple pattern count(\# Non Quant TP): } denotes the 	triple patterns in a query whose objects are \textit{not} arguments to the ranking function. This feature tests the robustness of a top-$k$ RDF system for 				processing non-quantifiable query patterns along with quantifiable query 				patterns.

\item \textbf{Count of joins(\# Joins): } This feature highlights the number of variables which are either 
\begin{inparaenum} [(1)]
	\item subject($S$), 
	\item object($O$), or
	\item fact identifiers($R$) 
\end{inparaenum}
of multiple triple patterns in a SPARQL query.

\item \textbf{Join type: } The indexing schemes used by different RDF engines can lead to difference in performance on different join types. Therefore, our benchmark tests the systems on different join types viz., Subject-Subject(SS), Reification(aka Fact Identifier)-Subject(RS) and Object-Subject(OS).

\item \textbf{Join degree: } This feature tests the RDF systems on degree of joins -- i.e., number of triple patterns with the same variable. For exposition, a join degree of (2,2,2,2) signifies a query with four joins, each with two triple patterns connected to the join variable.
\end{itemize}

Additionally our benchmark contains important \textbf{statistical} features, which includes:
\begin{itemize}
\item \textbf{Result cardinalities: } This feature tests the efficiency of RDF systems based on the number of results generated by a SPARQL query, \textit{without} the LIMIT clause. 

\item \textbf{Selectivities of queries: } This feature highlights the performance of RDF engines based on the number of results pruned by a SPARQL query, without the LIMIT clause. Selectivity is defined as the ratio of number of tuples appearing in the result to the total number of tuples in the Cartesian product of facts satisfying the triple patterns in the query.
\end{itemize}

Table~\ref{fig:querychar} shows the chief features of the 8 benchmark queries 
which we have created for both LGD and YAGO3 datasets. While YAGO3 consists of 
only point objects, LGD consists of different types of spatial objects and 
hence, the benchmark for LGD is more diverse comprising spatial joins of 
different object types. For both data sets, it is noteworthy that these queries 
provide a comprehensive coverage of all key features. For simplicity, we 
present results for linear ranking function, although our framework is aimed at 
convex monotonic functions.

\section{Experimental Results} \label{sec:experimentalResults}
In this section, we present a comprehensive experimental evaluation of \streak 
along the following dimensions: 
\begin{inparaenum} [(A)]
	\item the performance of spatial filter processing. Within this, we study (i) the impact of sideways
information passing and (ii) the performance of using the \squad based join over synchronous R-tree traversal based join.
	\item end-to-end holistic performance combining all features of \streak 
	\item the advantages of early-out top-$k$ processing in \streak.
\end{inparaenum}
Unless stated otherwise the results we present in this section correspond to a default setting of $k = 100$. 

\subsection{Performance of Spatial Join Processing in \streak}
\subsubsection{Impact of Sideways Information Passing and Optimal Node-Selection Algorithm} 
\label{subsec:DepthEst}

The node selection algorithm determines the optimal set of nodes of the \squad
 in order to determine the \textit{I-Range} and \textit{E-list} 
objects, which subsequently helps filter and narrow down the \textit{driven} 
sub-query in sideways information passing style. Note that we use \squad by keeping it fully memory-resident (although it is serialized to disk). Therefore we study the impact of node selection algorithm only under warm-cache setting. Figure~\ref{fig:sip_depthEstimation} shows the performance with and without the use of sideways information passing (SIP), powered by the node selection algorithm over \squad. Specifically, these plots show the time, in milliseconds plotted along Y-axis in log-scale, taken for completing all the benchmark queries over YAGO3 and LGD datasets (X-axis). 

Looking at individual queries in Fig.~\ref{fig:sip_depthEstimation}, we 
observe considerable improvement in query execution performance for all queries over YAGO3 dataset 
as well as $Q_{6}$ \& $Q_{8}$ of LGD when SIP and node selection algorithms are enabled. This is due to the fact that these queries are highly selective at the spatial join (with ranking), making the use \textit{I-Range} and \textit{E-list} objects for filtering the intermediate results from the driven plan highly effective. It can reduce the query processing time by upto \emph{3 orders of magnitude}. On the other hand, queries $Q_{1}$ to $Q_{5}$ in LGD have low selectivity at the spatial join operator (with ranking), with very little impact of SIP in skipping irrelevant entries from the driven plan. 

\begin{figure}[tb]
	\begin{subfigure}[c]{0.9\columnwidth}
		\includegraphics[width=\textwidth]{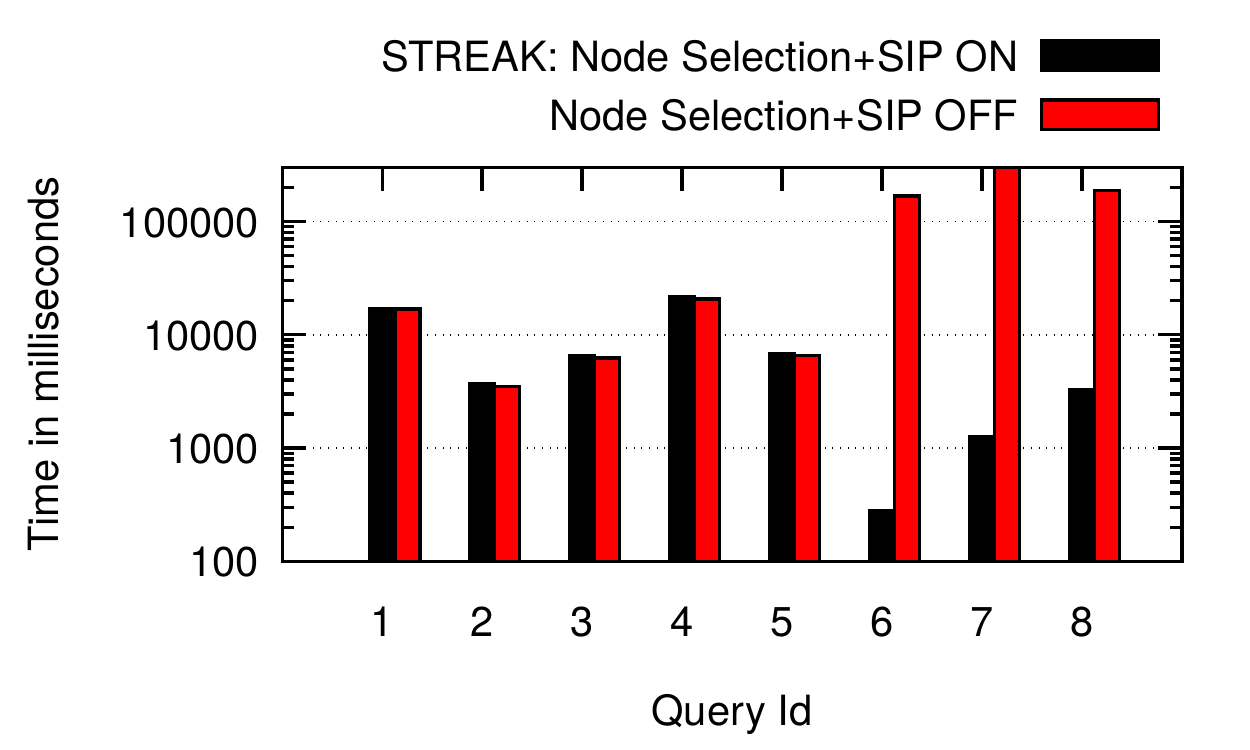}
						\caption*{\qquad LGD}			
	\end{subfigure}
	\begin{subfigure}[c]{0.9\columnwidth}
				\includegraphics[width=\textwidth]{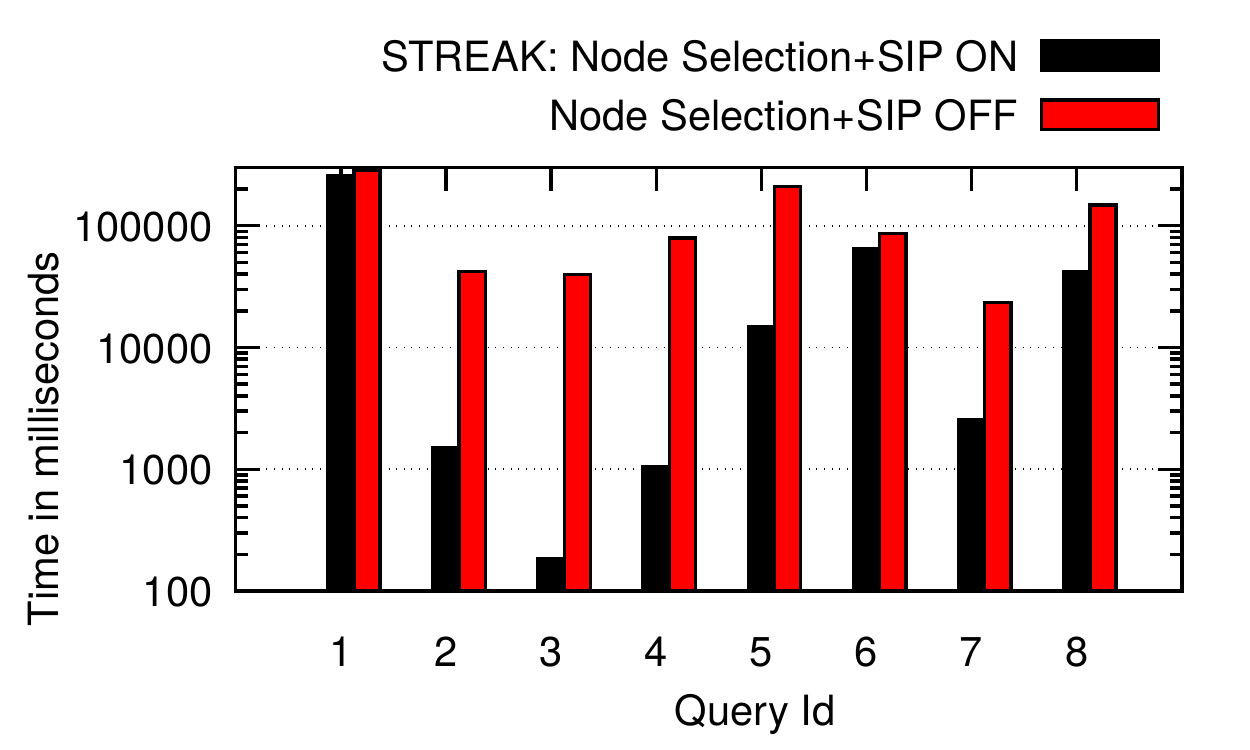}
			\caption*{\qquad YAGO3}
	\end{subfigure}
	\caption{Effect of Sideways Information Passing}
    \label{fig:sip_depthEstimation}
\end{figure}

\begin{figure}[tb]
	\begin{subfigure}[c]{0.9\columnwidth}
			\includegraphics[width=\textwidth]{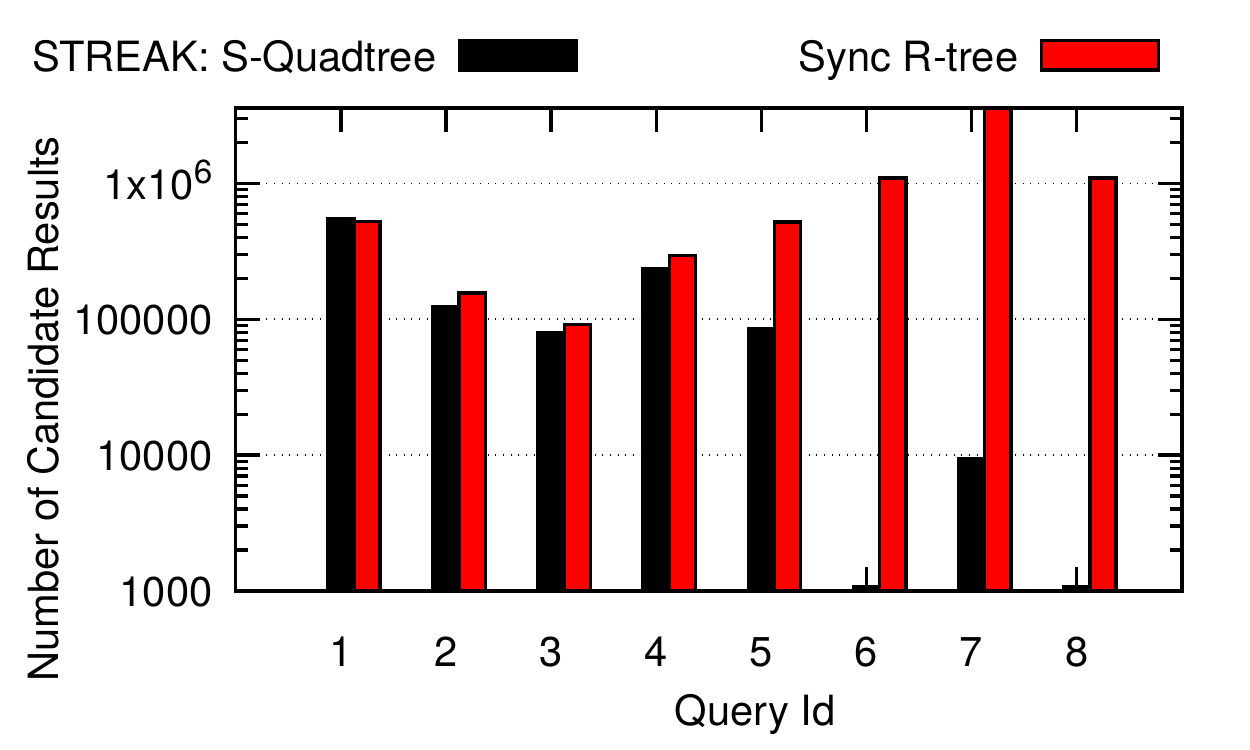}
			\caption*{\qquad LGD}		
	\end{subfigure}
	\begin{subfigure}[c]{0.9\columnwidth}
			\includegraphics[width=\textwidth]{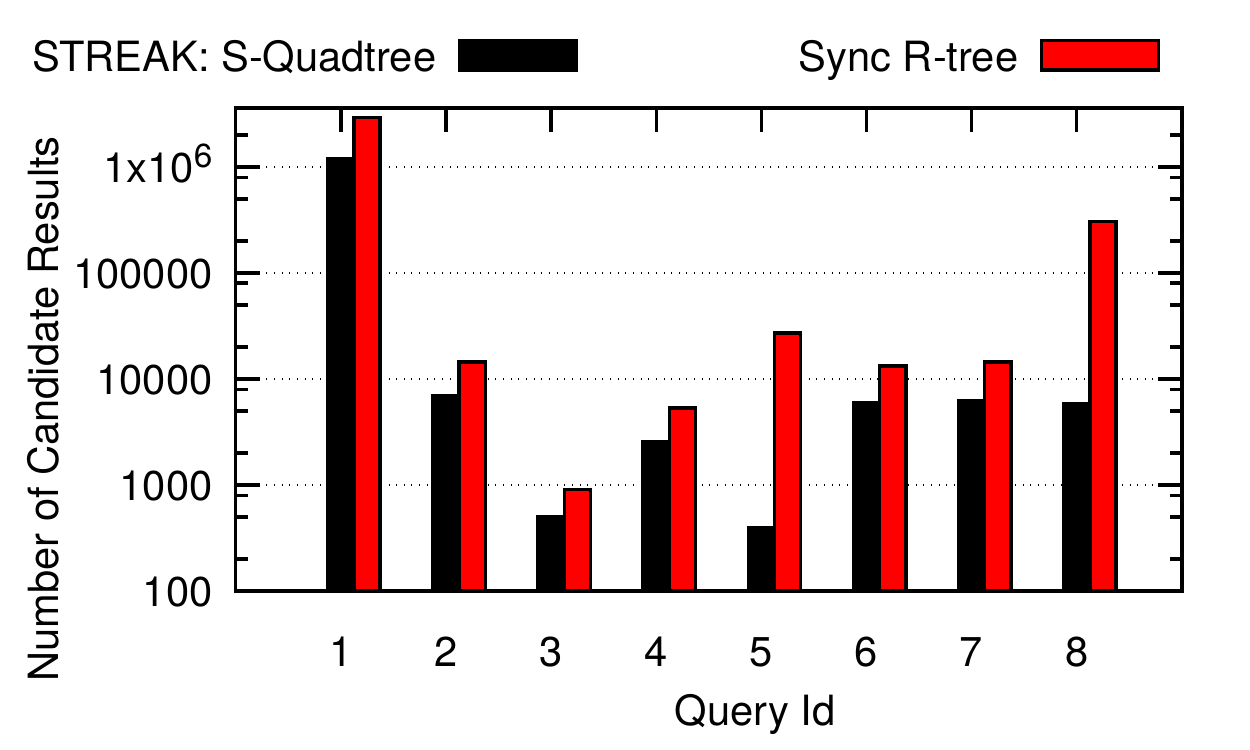}
			\caption*{\qquad YAGO3}
	\end{subfigure}
		\caption{\squad Vs. Sync. R-tree for Spatial Join}
		\label{fig:syncRTree_Squad}
\end{figure}

\subsection{Effectiveness of Adaptive Spatial Filter Processing} 
\label{subsec:AQP}
The goal of adaptively processing the spatial join queries using our APS algorithm is to select from either \textit{N-Plan} or \textit{S-Plan} alternatives as each block is processed. Recall that \textit{N-Plan} pushes down the numerical predicate evaluation, while \textit{S-Plan} pushes down the spatial joins. 
The first observation we can make from Fig~\ref{fig:aqpPerformance} is that 
\streak adaptively prefers \textit{S-Plan} for $Q_{2}$, $Q_{3}$, $Q_{4}$, $Q_{5}$  of LGD, and $Q_{2}$, $Q_{3}$, $Q_{4}$, $Q_{7}$ of YAGO3. Observe that the 
\textit{S-Plan} outperforms the other customized plan, 
\textit{N-Plan}, by more than an order of magnitude. This 
performance difference is due to the fact that for these queries almost all blocks from the right plan have to be fetched. 
\textit{S-Plan} is further accelerated in \streak by appropriately skipping irrelevant tuples in a sideways information passing style. The spatial objects 
present in \textit{I-Range} and \textit{E-list} of SQuad-tree's nodes, help in 
appropriately skipping irrelevant tuples.

\streak's customized \textit{N-Plan} helps it outperform the other 
customized \textit{S-Plan} for queries $Q_{6}$, $Q_{7}$, $Q_{8}$ of 
LGD, and $Q_{5}$, $Q_{8}$ of YAGO3. The performance gains is because for these queries only a small number of blocks from the right side plans have to be fetched, and hence the overhead of retrieving and uncompressing all blocks is greater in comparison to retrieving and uncompressing just a few blocks.

We may also observe that the time taken by \streak's AQP for $Q_{6}$ of YAGO3 is 
near the ballpark of \textit{S-Plan}, and is greater than 
\textit{N-Plan}. This is because for Query $Q_{6}$, both the 
\textit{S-Plan} and the \textit{N-Plan} fetch the same 
number of tuples. Now, \textit{S-Plan} does not have the 
additional overhead of random accesses. Hence, the running time of 
\textit{S-Plan} is correctly estimated by \streak to be lower, and 
thus it chooses it for query execution. But unfortunately because of underlying 
framework's (Quark-X) implementation, \textit{S-Plan} has to do many 
unsorted numerical predicate accesses, therefore \textit{S-Plan} is 
much slower than \textit{N-Plan}. If an extra index containing 
mapping from numerical identifier to its numerical value is added then we expect 
\textit{S-Plan} to outperform \textit{N-Plan}. 

\begin{figure}[tb]
	\begin{subfigure}[c]{0.9\columnwidth}
			\includegraphics[width=\textwidth]{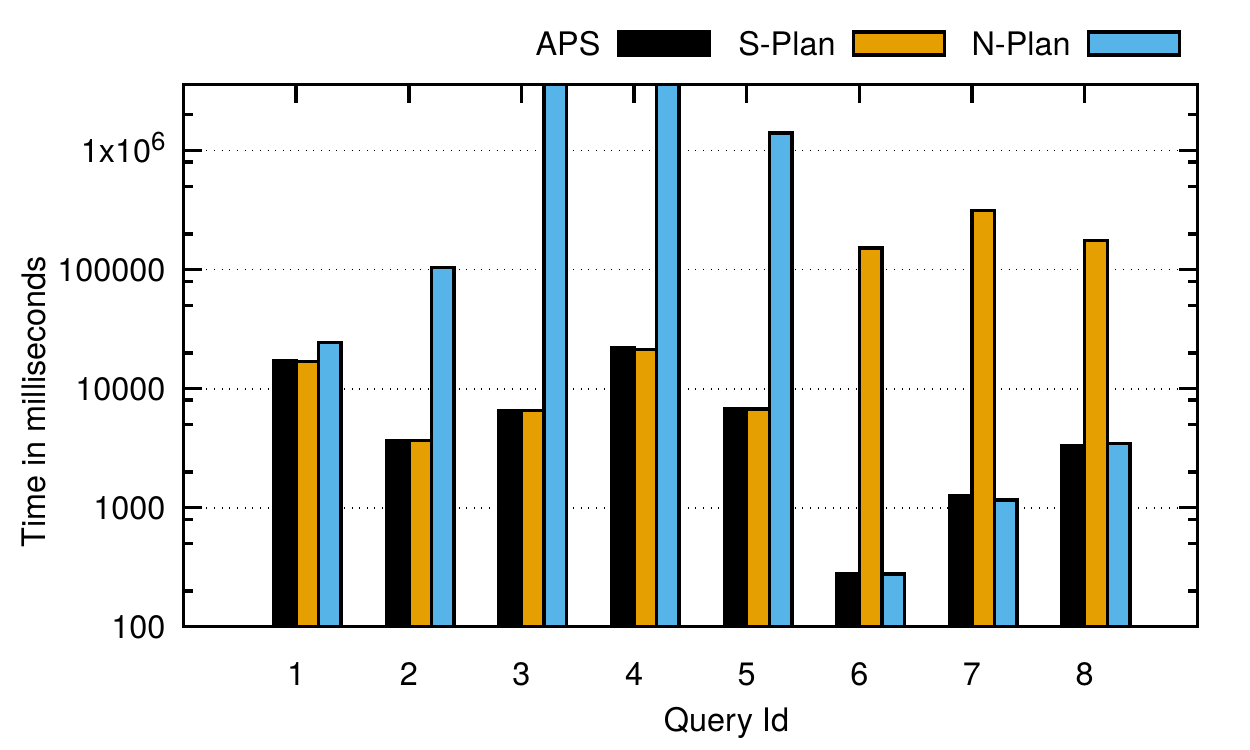}
			\caption*{\qquad LGD}		
	\end{subfigure}
	\begin{subfigure}[c]{0.9\columnwidth}
			\includegraphics[width=\textwidth]{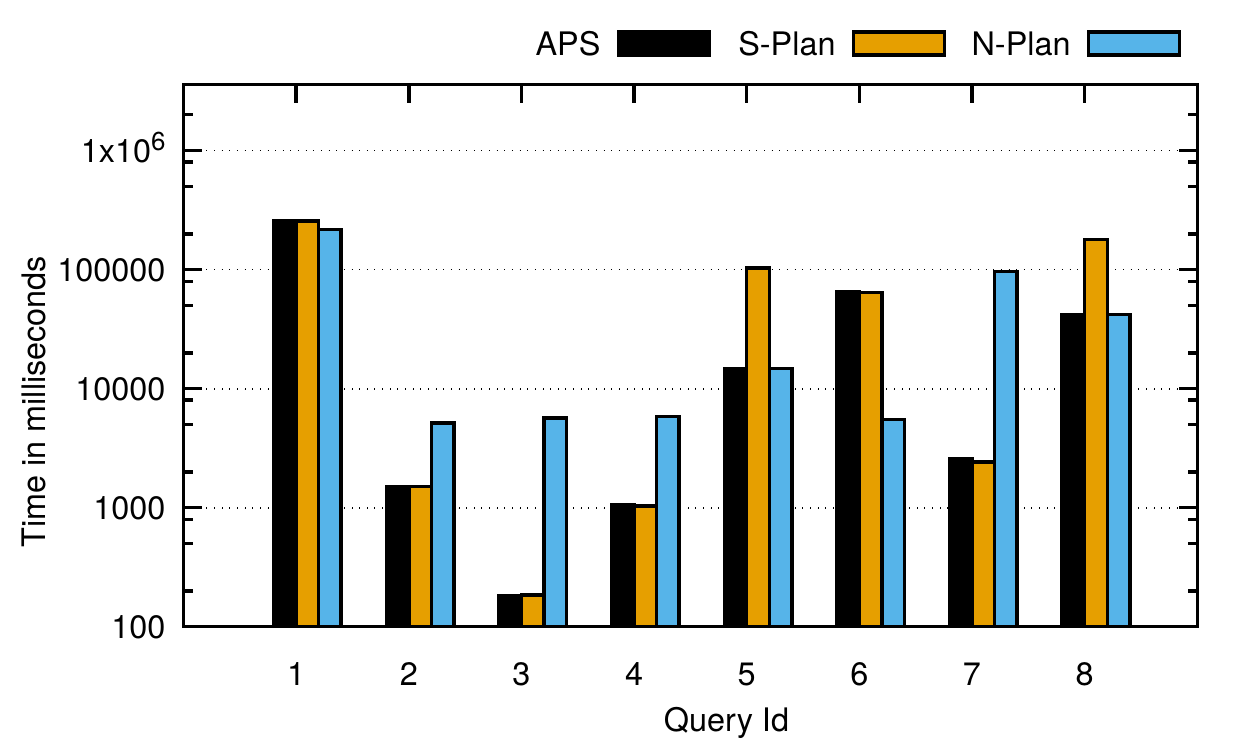}
			\caption*{\qquad YAGO3}
	\end{subfigure}
		\caption{Performace AQP versus N-Plan and S-Plan}
		\label{fig:aqpPerformance}
\end{figure}

\subsubsection{Comparison of Spatial Join Algorithms} 
\label{subsec:spatialJoinAlgorithms}
In the next set of experiments, we compare the performance of spatial join processing using \squad as opposed to using synchronous R-tree traversal based spatial join algorithm. Towards this, we incorporated the state of the art implementation of synchronous R-tree traversal join algorithm released by Sowell~et~al.~\cite{sowell2013experimental} into the \streak framework, and provided a run-time switch that can choose synchronous R-tree in place of \squad.

Figure~\ref{fig:syncRTree_Squad} plots the number of candidate results generated by the two algorithms (along Y-axis in log-scale) for all the benchmark queries. From these results it is immediately apparent that \squad generates smaller number of candidates, sometimes \emph{2 orders of magnitude} less --e.g., as in LGD $Q_6$ and YAGO3 $Q_5$-- than sync. R-tree based method. This can be attributed to the following advantages of \squad:
\begin{asparaitem}
	\item The use of \emph{characteristic set} information within the nodes of the \squad helps to reduce the number intermediate results significantly in comparison to standard spatial indexes such as R-trees.
	\item Encoded identifiers are able to achieve a limited granularity owing to the fixed size of identifiers. These identifiers when decoded back suffer from a loss of precision, which in-turn affects the filtering efficiency of spatial join algorithms such as Synchronous R-tree traversal. However, \streak owing to its use of \textit{I-Range} and \textit{E-list} objects is able to overcome this situation and is thus able to filter efficiently. 
	\item While both spatial join methods lack a well defined way of filtering the results in the driven plan as it is, the use of sideways information passing with \squad significantly reduces the cardinality of intermediate results. 
\end{asparaitem}
	 
\subsection{Comparison with Database Engines} 
\label{subsec:cmpDBEngines}
We now turn our attention to the comparison of end-to-end system performance of \streak with two state of the art publicly available database engines. Specifically, we compare with (i) \emph{PostgreSQL}~\cite{postgres}: a relational database system that has an excellent in-built support for spatial indexing via its generalized search tree (GiST) framework;  and (ii) \emph{Virtuoso}~\cite{virtuoso}: a state of the art commercial RDF processing engine that has a number of performance optimizations such as vectorization as well as cache-conscious processing specifically useful for RDF/SPARQL processing. Note that our current implementation of \streak has neither the vectorized processing nor the cache-conscious features, although these are part of our future work plans. 

Before we discuss the query processing performance, we briefly present the effective size of databases created by the three engines summarized in Table~\ref{tab:loadingtime}. As we already mentioned in Section~\ref{sec:eval}, Virtuoso could not be used for LGD dataset since it lacks support for varied geometry types such POINTS, POLYGONS and LINESTRINGS. At 11 GB of database size,  Virtuoso, which employs highly compressed encoding, is the most compact database for YAGO3. While \streak also has a comparable database size (14 GB), PostgreSQL on the other hand has a significantly larger database size --almost \textbf{5$\times$} larger. Similar observations also hold for LGD between \streak and PostgreSQL database sizes.

\begin{table}[!tb]
\caption{On-Disk Database Size for YAGO3 and LGD}
\label{tab:loadingtime}
\centering
\begin{tabular}{l  c  c}
\toprule
\textbf{Database Engines} & \textbf{YAGO3} & \textbf{LGD} \\
& Size (GB) & Size (GB) \\
\midrule
\textbf{\streak} & 14 & \textbf{6.4} \\
\textbf{PostgreSQL} & 54  & 22 \\
\textbf{Virtuoso} & \textbf{11} & ---NA---\\
\bottomrule
\end{tabular}
\end{table}

We summarize the comparative query processing performance of the three engines in Figure~\ref{fig:postgres_virtuoso}. Due to lack of space, we only show and discuss the warm-cache performance numbers of the three database engines. Their performance in cold-cache setting has similar behavior and is shown in Figure~\ref{fig:postgres_virtuosoCold}. The Y-axis in these plots corresponds to the end-to-end wall-clock time in milliseconds for each query, and is plotted in log-scale. 

From these results, we can observe that \streak outperforms both PostgreSQL and Virtuoso on all queries, by a significant margin. For a number of queries --\textit{viz.}, $Q_6 - Q_8$ of LGD and $Q_5, Q_7, Q_8$ of YAGO3 -- PostgreSQL could not complete within its allotted time. Virtuoso failed to complete for $Q_5$ of YAGO3 as well. The poor performance of PostgreSQL can be attributed to its preference to choose nested loop joins, due to its relatively weak cost models for RDF. For Virtuoso, however, we cannot comment on its poor performance as it is closed source and no published work describes its functionality and query optimization techniques.

 \begin{figure}
 	\centering
	\begin{subfigure}[b]{0.9\columnwidth}
		\includegraphics[width=\columnwidth]{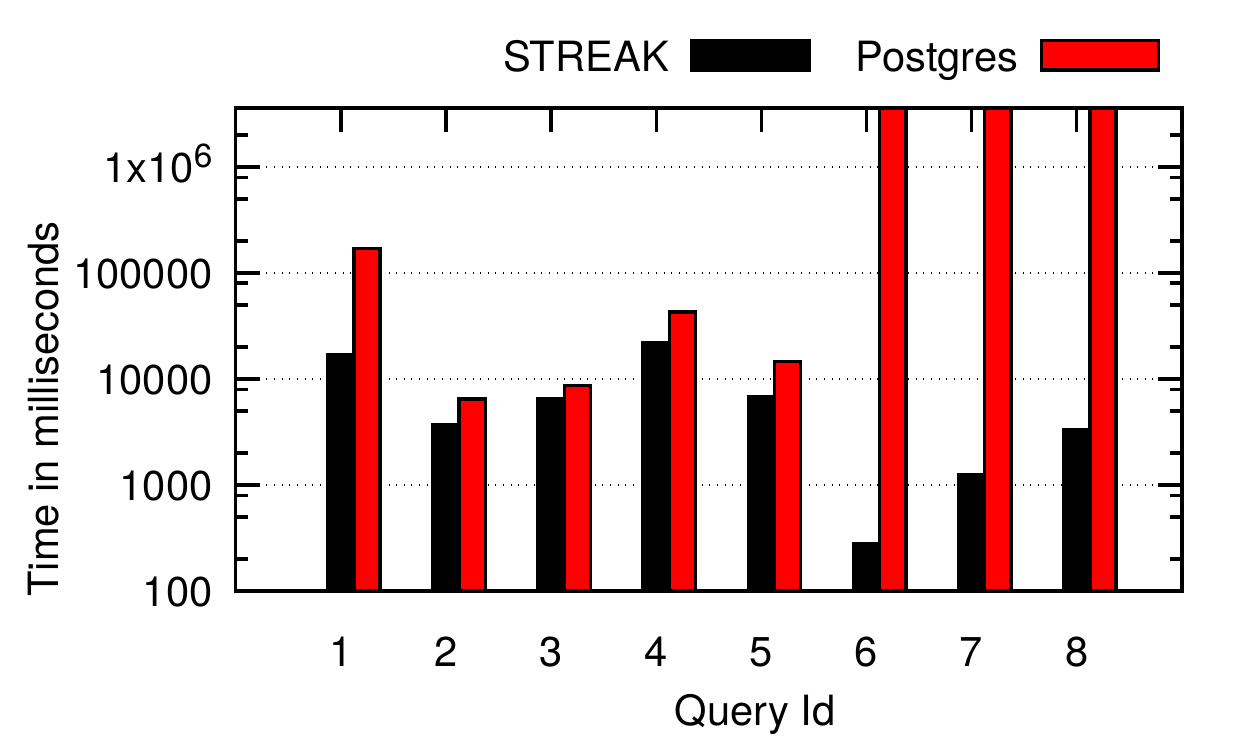}
	\caption*{LGD}
	\end{subfigure}
	\begin{subfigure}[b]{0.9\columnwidth}	\includegraphics[width=\columnwidth]{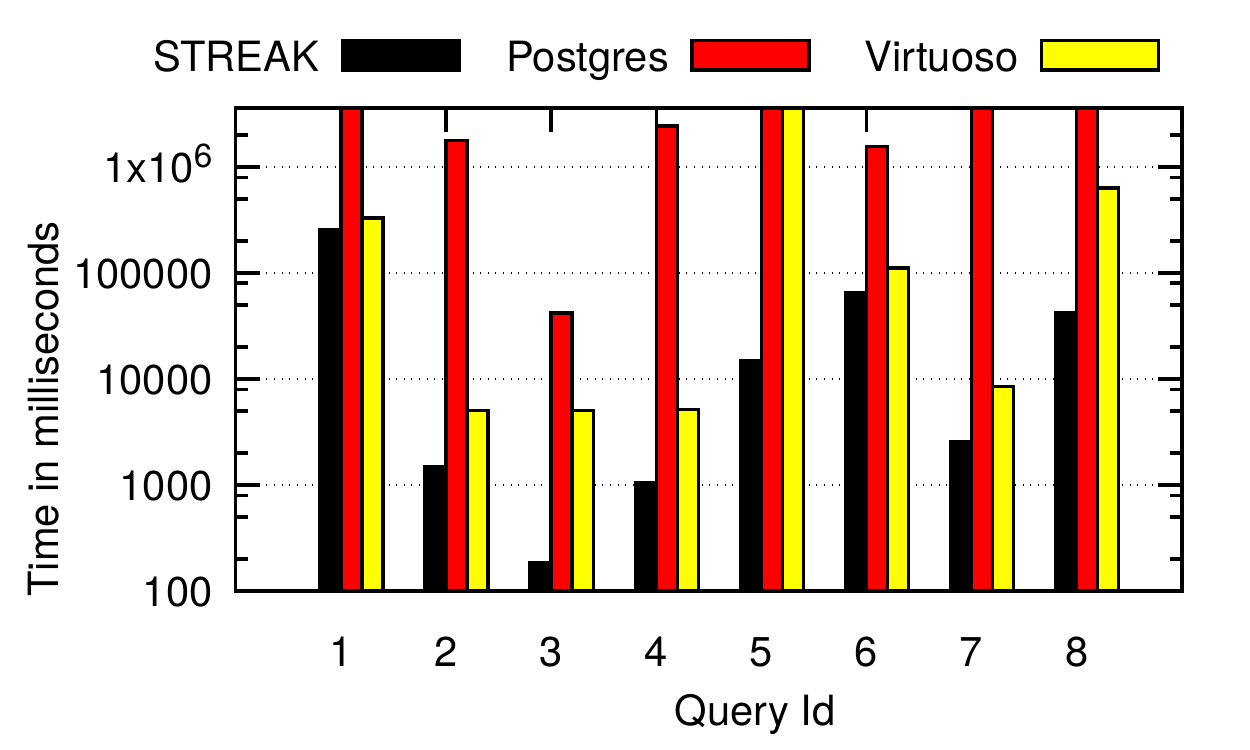}
	\caption*{YAGO3}
	\end{subfigure}
	\caption{Performance of \streak Vs. PostgreSQL and Virtuoso \textit{in Warm Cache}}
	\label{fig:postgres_virtuoso}
\end{figure}

 \begin{figure}
 	\centering
	\begin{subfigure}[b]{0.9\columnwidth}
		\includegraphics[width=\columnwidth]{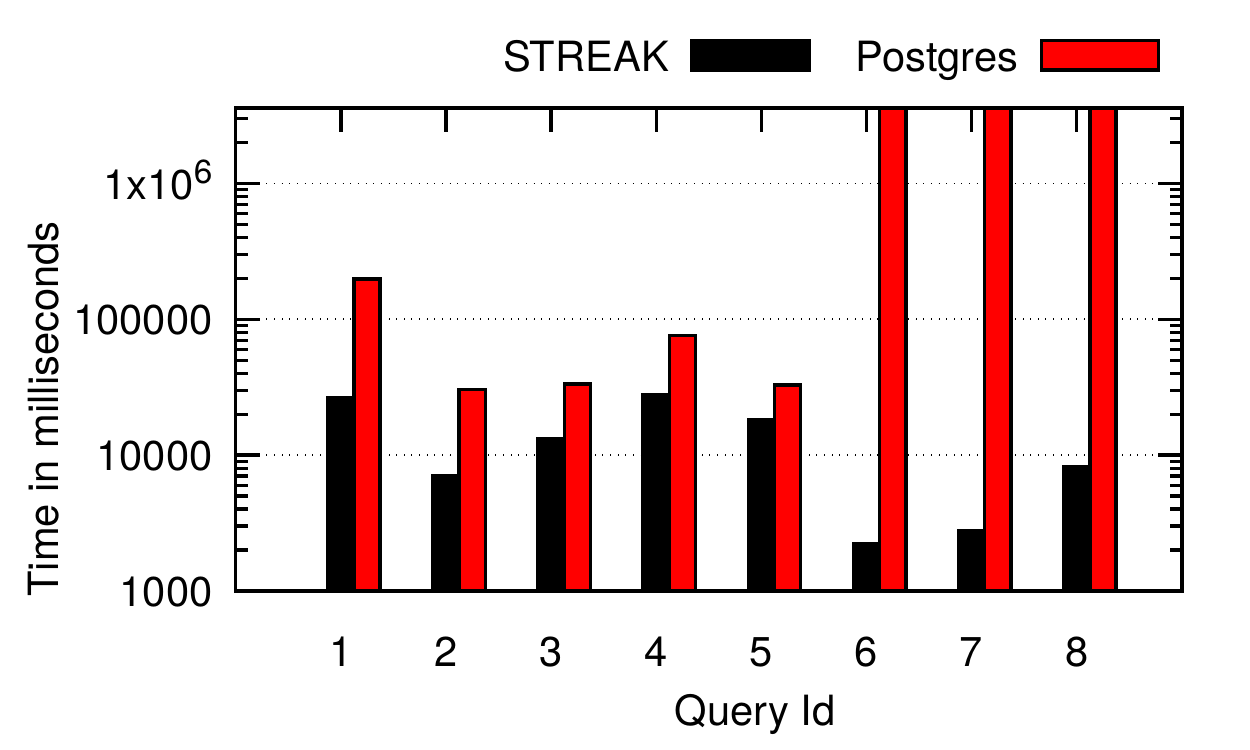}
	\caption*{LGD}
	\end{subfigure}
	\begin{subfigure}[b]{0.9\columnwidth}	\includegraphics[width=\columnwidth]{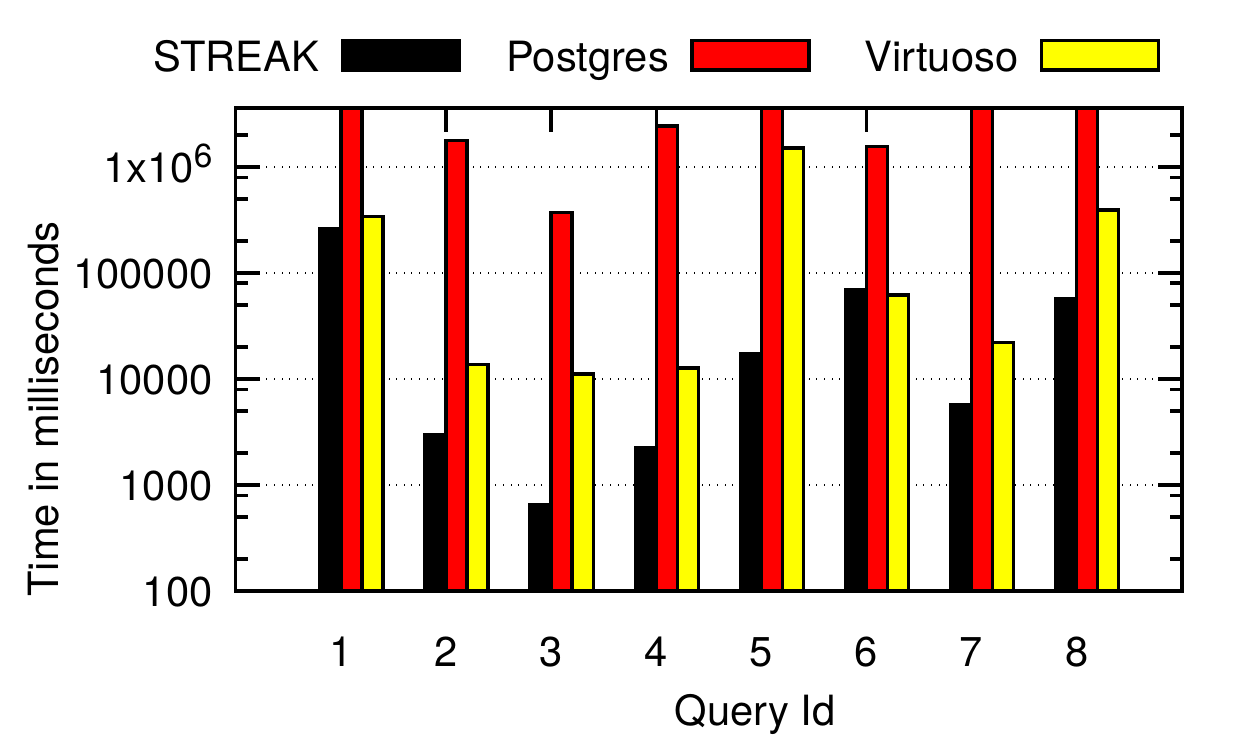}
	\caption*{YAGO3}
	\end{subfigure}
	\caption{Performance of \streak Vs. PostgreSQL and Virtuoso \textit{in Cold Cache}}
	\label{fig:postgres_virtuosoCold}
\end{figure}

\subsection{Comparison with varying $k$} 
\label{subsec:withvaryingK}
Finally, we briefly present the performance evaluation when $k$ value is varied. Figure~\ref{fig:varyk} plots the geometric mean of runtime of all queries in the benchmark for $k = \left\{ 1, 10, 50, 100\right\}$ in warm-cache scenario for LGD. Results of cold-cache setting and YAGO3 benchmark could not be presented due to lack of space, but show similar behavior. One can observe that \streak outperforms PostgreSQL by a large margin for all values $k$. PostgreSQL is unaffected by the variation in the value of $k$ since it simply performs the full evaluation followed by sort in all cases. 

It is noteable to see the performance of our APS algorithm in comparison to fixed plans namely, \textit{N-Plan} and \textit{S-Plan}. For $k$ = 1, 10 the \textit{N-Plan} outperforms the \textit{S-Plan}; while for $k$ = 50, 100 the \textit{S-Plan} outperforms the \textit{N-Plan}. This is because of the need for scanning fewer blocks for lower values of $k$, thus making \textit{N-Plan} a cheap option. While many blocks need to be scanned for larger values of $k$, thus making \textit{S-Plan} a better option. It is noteworthy that \streak shows outstanding performance in comparison to both \textit{S-Plan} and \textit{N-Plan}, owing to better plan selection and switching at run-time.
\begin{figure}
	\centering
	  \includegraphics[width=0.9\columnwidth]{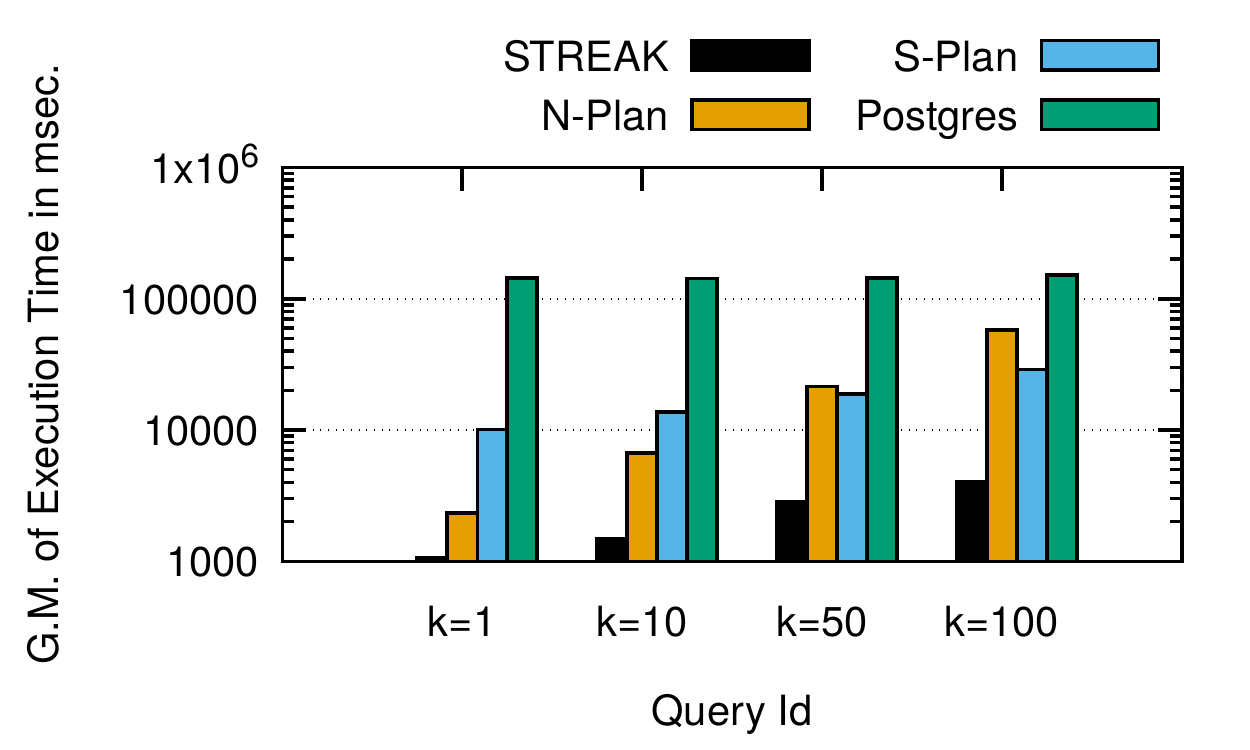}
	  \caption*{LGD on Warm Cache}
	  \caption{Performance with Varying $k$}
      \label{fig:varyk}
\end{figure}

\section{Related Work} \label{sec:relwork}
In this section, we provide a full overview of techniques that overlap with the
core aspects of \streak. These primarily relate to RDF spatial joins, top-$k$ joins, spatial top-$k$ joins and adaptive query
processing.

\noindent \textbf{Spatial Joins}\\
We now look at the work done specifically on spatial joins in both relational 
and RDF settings.

Parliament~\cite{parliament2012} is a spatial RDF engine built on top of
Jena~\cite{jenaTDB}. Parliament supports GeoSPARQL features and uses R-tree as the
spatial index. Additional techniques to enhance spatial queries, as adopted by
Geo-store~\cite{geoStore} and in RDF engine proposed by Liagouris et
al.~\cite{encodingSpatialRDF} involve encoding spatial data by employing
Hilbert space-filling curves for preserving spatial locality. 

Another framework, S-Store~\cite{sStore}, extends state-of-the-art RDF store
gStore~\cite{zou2011gstore}, by indexing data based on their structure.
Subsequently, queries are accelerated by pruning based on both spatial and
semantic constraints in the query. Yet another spatiotemporal storage
framework is g$^{st}$-Store~\cite{wang2014gst}, which is an extension of
gStore~\cite{zou2011gstore}. It processes spatial queries using its tree-style
index ST-tree, with a top-down search algorithm. 

However, the above frameworks, while capable of handling spatial filters in the
query, are not directly amenable for either top-$k$ processing or adaptive
querying. They miss out features to early-prune intermediate results and to
process unsorted orders on spatial attributes (imposed by top-$k$ query
processing), which are critical for optimizing the workloads seen here.
\streak in stark contrast provides efficient ways of supporting top-$k$ spatial
joins using its customized \squad and query processing algorithms. This
helps \streak outperform state-of-the-art systems like Virtuoso by a large
margin as shown experimentally in sub-section~\ref{subsec:cmpDBEngines}.

Relational approaches such as Synchronous R-tree traversal~\cite{syncRTree} and
TOUCH~\cite{nobari2013touch} apply spatial-join algorithm on hierarchical tree-like data structure. TOUCH is an in-memory based technique, which uses R-tree traversal for evaluating spatial
joins. \streak differs from TOUCH in using identifier encoding and
semantic information embedded within the datasets for early-pruning,  \jyoticomment{ and for improving performance.}and by
utilizing statistics stored within the spatial index to choose the query plan
at runtime for improving performance. Sync. R-tree traversal builds
R-tree indexes over the datasets participating in the spatial join.
Starting from the root of the trees, the two datasets are then synchronously
traversed to check for intersections, with the join happening at the leaf
nodes. Inner node overlap and dead space are two shortcomings in this
technique, which are addressed with \streak's \squad that uses
space-oriented partitioning. We compare \streak with sync. R-tree
traversal in Section~\ref{sec:experimentalResults}.

\noindent
\textbf{Top-$k$ Joins}\\
Hash-based Rank Join (HRJN)~\cite{ilyas2004rank} and Nested Loop Rank Join 
(NRJN)~\cite{ilyas2004supporting} are two widely used top-$k$ join algorithms 
in the relational world. HRJN accesses objects from left (or right) side of join 
and joins them with tuples seen so-far from right (or left) side of join. All 
join results are fed to a priority queue, which outputs the results to the users 
if they are above threshold, thus producing results incrementally. Nested Loop 
Rank Join (NRJN) algorithm is similar to HRJN, except that it follows a nested 
loop join strategy instead of buffering join inputs.  We do not compare \streak 
with a general top-$k$ join algorithms, HRJN and NRJN, since such top-$k$ 
operators for such spatial workloads have been shown to perform poorly when compared to a block-based 
algorithm~\cite{TopKSpatialJoins} and instead compare with the state-of-the-art 
spatial top-$k$ algorithm.

Quark-X~\cite{quarkX} is a recently proposed top-$k$ query processing engine,
which uses summarized numerical indexes for retrieving and ranking RDF facts.
Quark-X showed experimentally that it outperformed the other top-$k$ query
processing systems by a large margin. However, Quark-X is suited only for
1-dimensional data, as its techniques can not be applied for multi-dimensional
data (sorted order is a requisite for Quark-X). Hence, its approach leads to
many unnecessary comparison operations in the non-sorted dimension. \streak in
contrast reduces the number of comparisons required drastically, by using its
schema-aware \squad. 

\noindent
\textbf{Spatial Top-$k$ Joins}\\
Returning the top-$k$ results on spatial queries, where the final results are
ranked by using a distance metric has been already studied by Ljosa, et
al.~\cite{topKSpatialJoinProbabilistic}. However, such ranking mechanism often
have to be restricted to specific, pre-defined aggregation functions. In
comparison, \streak attempts to solve a more challenging problem with the
spatial constraint specified in the join predicate.

To the best of our knowledge, the work closely related to \streak is done by
\textit{Qi et al.,} ~\cite{TopKSpatialJoins}. In this approach \textit{Qi et
al.,} retrieve blocks of data ordered by object scores, from input datasets, which are then spatially joined. Their approach:
(1) requires data to be sorted based on the ranking function, making it
unsuitable for arbitrary user-defined ranking function. \streak, in contrast,
does not impose any such limitation on the data. (2) only supports spatial
joins and would require non-trivial modifications for supporting complex
patterns seen in SPARQL queries~\cite{weiss2008hexastore, neumann2008rdf,
neumann2009scalable}. \streak in comparison evaluates complex queries
efficiently. (3) uses \textit{only} block-wise retrieval for joining.
Block-wise retrieval is expected to be fast only if the results are found in
first few block accesses\jyoticomment{ from driver and driven sub-queries}. However, such an
approach suffers high overheads when blocks deep in the sorted collections need
to be retrieved. \jyoticomment{-- we validate this observation experimentally in
Section~\ref{sec:experimentalResults}}\streak circumvents this drawback by
adaptively switching between its customized plans at run-time (described in
sub-section~\ref{subsec:aqp}). (4) incur additional index creation cost at runtime, as
they build spatial indexes at run-time. \streak, on the other hand, builds its
spatial index \squad during pre-processing stage, thus incurring zero index creation overhead during query processing.

\noindent
\textbf{Adaptive Query Processing}\\ 
Eddies~\cite{avnur2000eddies} and Content Based Retrieval (CBR)
~\cite{bizarro2005content} are two state-of-the-art approaches which explored AQP for switching plans during query execution. Both these approaches use random tuples for
profiling, relying on a machine-learning based solution to model routing
predictions. Our work directly contrasts ML-based approaches by using
statistics for each \textit{block} of spatial data, which helps in determining
the selectivity of spatial join operator. By using fine-grained block-level
statistics, we mitigate the errors that are typically associated with a
model-based approach, especially when there are many joins involved.
Additionally, unlike CBR, which has high routing and learning
overhead, our spatial AQP algorithm incurs a very small routing overhead, owing
to cost estimate calculations for only the customized plans. We validated this argument and found the overhead of AQP to be very small --- just 5-10$\%$ of the overall overhead. Finally, routing predictions of Eddy and CBR depend on the model used and the size of the training datasets, while our statistics based approach does not require modeling and is relatively unaffected by the training dataset size.

\section{Conclusion} \label{sec:concl}
This paper introduces \streak, an efficient top-$k$ query processing engine with spatial filters. The novel contributions of this work are its soft-schema aware \squad spatial index, new spatial join algorithm which prunes the search space using its optimal node selection algorithm, and adaptive query plan generation algorithm called \textit{APS} which adaptively switches query plan at runtime with very low overhead. We show experimentally that \streak outperforms popular spatial join algorithms as well state-of-the-art end-to-end engines like PostgreSQL and Virtuoso. As part of future work, we plan to implement top-$k$ RDF reasoner with spatial filters.

{
\bibliographystyle{ACM-Reference-Format}
\bibliography{sample-bibliography} 
}

\section{Appendix}
\subsection{LGD Queries}
\begin{lstlisting}[escapechar=!]
!\bfseries{\normalsize{\textcolor{black}{Query 1:}}}!
BASE <http://yago-knowledge.org/resource/> .
PREFIX rdf: <http://www.w3.org/1999/02/22-rdf-syntax-ns#>
SELECT ?place ?nplace ?typePred1 ?typePred2  
WHERE {
  ?r rdf:subject ?place. 
  ?r rdf:predicate ?typePred1. 
  ?r rdf:object <http://geoknow.eu/uk_points#hotel>. 
  ?r <hasConfidence> ?conf. 
  ?place <http://yago-knowledge.org/resource/hasGeometry> ?long. 
  ?r1 rdf:subject ?nplace. 
  ?r1 rdf:predicate ?typePred2. 
  ?r1 rdf:object <http://geoknow.eu/uk_natural#park>. 
  ?r1 <hasConfidence> ?conf1. 
  ?nplace <http://yago-knowledge.org/resource/hasGeometry> ?nlong. 
  FILTER((?long, ?nlong) < 50)
} ORDER BY ASC(?conf + ?conf1) LIMIT k

!\bfseries{\normalsize{\textcolor{black}{Query 2:}}}!
BASE <http://yago-knowledge.org/resource/> .
PREFIX rdf: <http://www.w3.org/1999/02/22-rdf-syntax-ns#>
SELECT ?place ?nplace ?typePred1 ?typePred2 
WHERE {
  ?r rdf:subject ?place. 
  ?r rdf:predicate ?typePred1. 
  ?r rdf:object <http://geoknow.eu/uk_points#hotel>. 
  ?r <hasConfidence> ?conf. 
  ?place <http://yago-knowledge.org/resource/hasGeometry> ?long. 
  ?r1 rdf:subject ?nplace. 
  ?r1 rdf:predicate ?typePred2. 
  ?r1 rdf:object <http://geoknow.eu/uk_points#police>. 
  ?r1 <hasConfidence> ?conf1. 
  ?nplace <http://yago-knowledge.org/resource/hasGeometry> ?nlong. 
  FILTER((?long, ?nlong) < 50)
} ORDER BY ASC(?conf + ?conf1) LIMIT k

!\bfseries{\normalsize{\textcolor{black}{Query 3:}}}!
BASE <http://yago-knowledge.org/resource/> .
PREFIX rdf: <http://www.w3.org/1999/02/22-rdf-syntax-ns#>
SELECT ?place ?nplace ?typePred1 ?typePred2 ?label1 ?name2 
WHERE {
  ?r rdf:subject ?place. 
  ?r rdf:predicate ?typePred1. 
  ?r rdf:object <http://geoknow.eu/uk_points#hotel>. 
  ?r <hasConfidence> ?conf. 
  ?place <http://yago-knowledge.org/resource/hasGeometry> ?long. 
  ?place <http://www.w3.org/2000/01/rdf-schema#label> ?label1. 
  ?r1 rdf:subject ?nplace. 
  ?r1 rdf:predicate ?typePred2. 
  ?r1 rdf:object <http://geoknow.eu/uk_points#police>. 
  ?r1 <hasConfidence> ?conf1. 
  ?nplace <http://yago-knowledge.org/resource/hasGeometry> ?nlong. 
  ?nplace <http://geoknow.eu/uk_points#name> ?name2. 
  FILTER((?long, ?nlong) < 50)
} ORDER BY ASC(?conf + ?conf1) LIMIT k

!\bfseries{\normalsize{\textcolor{black}{Query 4:}}}!
BASE <http://yago-knowledge.org/resource/> .
PREFIX rdf: <http://www.w3.org/1999/02/22-rdf-syntax-ns#>
SELECT ?place ?nplace ?typePred1 ?typePred2 ?label1 ?name1 ?name2 
WHERE {
  ?r rdf:subject ?place. 
  ?r rdf:predicate ?typePred1. 
  ?r rdf:object <http://geoknow.eu/uk_points#pub>. 
  ?r <hasConfidence> ?conf. 
  ?place <http://yago-knowledge.org/resource/hasGeometry> ?long. 
  ?place <http://www.w3.org/2000/01/rdf-schema#label> ?label1. 
  ?place <http://geoknow.eu/uk_points#name> ?name1. 
  ?r1 rdf:subject ?nplace. 
  ?r1 rdf:predicate ?typePred2. 
  ?r1 rdf:object <http://geoknow.eu/uk_points#police>. 
  ?r1 <hasConfidence> ?conf1. 
  ?nplace <http://yago-knowledge.org/resource/hasGeometry> ?nlong. 
  ?nplace <http://geoknow.eu/uk_points#name> ?name2. 
  FILTER((?long, ?nlong) < 50)
} ORDER BY ASC(?conf + ?conf1) LIMIT k

!\bfseries{\normalsize{\textcolor{black}{Query 5:}}}!
BASE <http://yago-knowledge.org/resource/> .
PREFIX rdf: <http://www.w3.org/1999/02/22-rdf-syntax-ns#>
SELECT ?place ?nplace ?typePred1 ?typePred2 ?label1 ?name1 ?name2 ?label2 
WHERE {
  ?r rdf:subject ?place. 
  ?r rdf:predicate ?typePred1. 
  ?r rdf:object <http://geoknow.eu/uk_natural#park>. 
  ?r <hasConfidence> ?conf. 
  ?place <http://yago-knowledge.org/resource/hasGeometry> ?long. 
  ?place <http://www.w3.org/2000/01/rdf-schema#label> ?label1. 
  ?r1 rdf:subject ?nplace. 
  ?r1 rdf:predicate ?typePred2. 
  ?r1 rdf:object <http://geoknow.eu/uk_points#police>. 
  ?r1 <hasConfidence> ?conf1. 
  ?nplace <http://yago-knowledge.org/resource/hasGeometry> ?nlong. 
  ?nplace <http://geoknow.eu/uk_points#name> ?name2. 
  ?nplace <http://www.w3.org/2000/01/rdf-schema#label> ?label2.
  FILTER((?long, ?nlong) < 50)
} ORDER BY ASC(?conf + ?conf1) LIMIT k

!\bfseries{\normalsize{\textcolor{black}{Query 6:}}}!
BASE <http://yago-knowledge.org/resource/> .
PREFIX rdf: <http://www.w3.org/1999/02/22-rdf-syntax-ns#>
SELECT ?place ?nplace ?typePred1 ?typePred2 
WHERE {
  ?r rdf:subject ?place. 
  ?r rdf:predicate ?typePred1. 
  ?r rdf:object <http://geoknow.eu/uk_natural#park>. 
  ?r <hasConfidence> ?conf. 
  ?place <http://yago-knowledge.org/resource/hasGeometry> ?long. 
  ?r1 rdf:subject ?nplace. 
  ?r1 rdf:predicate ?typePred2. 
  ?r1 rdf:object <http://geoknow.eu/uk_roads#roads>. 
  ?r1 <hasConfidence> ?conf1. 
  ?nplace <http://yago-knowledge.org/resource/hasGeometry> ?nlong. 
  FILTER((?long, ?nlong) < 50)
} ORDER BY ASC(?conf + ?conf1) LIMIT k

!\bfseries{\normalsize{\textcolor{black}{Query 7:}}}!
BASE <http://yago-knowledge.org/resource/> .
PREFIX rdf: <http://www.w3.org/1999/02/22-rdf-syntax-ns#>
SELECT ?place ?nplace ?typePred1 ?typePred2 
WHERE {
  ?r rdf:subject ?place. 
  ?r rdf:predicate ?typePred1. 
  ?r rdf:object <http://geoknow.eu/uk_points#hotel>. 
  ?r <hasConfidence> ?conf. 
  ?place <http://yago-knowledge.org/resource/hasGeometry> ?long. 
  ?r1 rdf:subject ?nplace. 
  ?r1 rdf:predicate ?typePred2. 
  ?r1 rdf:object <http://geoknow.eu/uk_roads#roads>. 
  ?r1 <hasConfidence> ?conf1. 
  ?nplace <http://yago-knowledge.org/resource/hasGeometry> ?nlong. 
  FILTER((?long, ?nlong) < 50)
} ORDER BY ASC((++0+1*?conf/1) (++0+1*?conf1/1)) LIMIT k

!\bfseries{\normalsize{\textcolor{black}{Query 8:}}}!
BASE <http://yago-knowledge.org/resource/> .
PREFIX rdf: <http://www.w3.org/1999/02/22-rdf-syntax-ns#>
SELECT ?place ?nplace ?typePred1 ?typePred2 ?label2 
WHERE {
  ?r rdf:subject ?place. 
  ?r rdf:predicate ?typePred1. 
  ?r rdf:object <http://geoknow.eu/uk_natural#park>. 
  ?r <hasConfidence> ?conf. 
  ?place <http://yago-knowledge.org/resource/hasGeometry> ?long. 
  ?r1 rdf:subject ?nplace. 
  ?r1 rdf:predicate ?typePred2. 
  ?r1 rdf:object <http://geoknow.eu/uk_roads#roads>. 
  ?r1 <hasConfidence> ?conf1. 
  ?nplace <http://yago-knowledge.org/resource/hasGeometry> ?nlong. 
  ?nplace <http://www.w3.org/2000/01/rdf-schema#label> ?label2. 
  FILTER((?long, ?nlong) < 50)
} ORDER BY ASC(?conf + ?conf1) LIMIT k

\end{lstlisting}

\subsection{YAGO Queries}
\begin{lstlisting}[escapechar=!]
!\bfseries{\normalsize{\textcolor{black}{Query 1:}}}!
BASE <http://yago-knowledge.org/resource/> .
PREFIX rdf: <http://www.w3.org/1999/02/22-rdf-syntax-ns#>
SELECT ?place ?nplace ?loc1 ?loc2 
WHERE {
  ?place <http://yago-knowledge.org/resource/hasPopulationDensity> ?popul. 
  ?place <http://yago-knowledge.org/resource/hasGeometry> ?long. 
  ?place <http://yago-knowledge.org/resource/isLocatedIn> ?loc1. 
  ?nplace <http://yago-knowledge.org/resource/isLocatedIn> ?loc2. 
  ?nplace <http://yago-knowledge.org/resource/hasGeometry> ?nlong. 
  ?nplace <http://yago-knowledge.org/resource/hasNumberOfPeople> ?popul1. 
  FILTER((?long, ?nlong) < 50)
} ORDER BY ASC(?popul + ?popul1) LIMIT k

!\bfseries{\normalsize{\textcolor{black}{Query 2:}}}!
BASE <http://yago-knowledge.org/resource/> .
PREFIX rdf: <http://www.w3.org/1999/02/22-rdf-syntax-ns#>
SELECT ?place ?nplace ?loc1 ?loc2 
WHERE {
  ?place <http://yago-knowledge.org/resource/hasPopulationDensity> ?popul. 
  ?place <http://yago-knowledge.org/resource/hasEconomicGrowth> ?ecoGrowth. 
  ?place <http://yago-knowledge.org/resource/hasGeometry> ?long. 
  ?place <http://yago-knowledge.org/resource/isLocatedIn> ?loc1. 
  ?nplace <http://yago-knowledge.org/resource/isLocatedIn> ?loc2. 
  ?nplace <http://yago-knowledge.org/resource/hasGeometry> ?nlong. 
  ?nplace <http://yago-knowledge.org/resource/hasNumberOfPeople> ?popul1. 
  FILTER((?long, ?nlong) < 50)
} ORDER BY ASC(?popul + ?popul1 + ?ecoGrowth) LIMIT k

!\bfseries{\normalsize{\textcolor{black}{Query 3:}}}!
BASE <http://yago-knowledge.org/resource/> .
PREFIX rdf: <http://www.w3.org/1999/02/22-rdf-syntax-ns#>
SELECT ?place ?nplace ?loc1 ?loc2 ?connection 
WHERE {
  ?place <http://yago-knowledge.org/resource/hasEconomicGrowth> ?ecoGrowth. 
  ?connection <http://yago-knowledge.org/resource/isConnectedTo> ?place. 
  ?place <http://yago-knowledge.org/resource/hasGeometry> ?long. 
  ?place <http://yago-knowledge.org/resource/isLocatedIn> ?loc1. 
  ?nplace <http://yago-knowledge.org/resource/isLocatedIn> ?loc2. 
  ?nplace <http://yago-knowledge.org/resource/hasGeometry> ?nlong. 
  ?nplace <http://yago-knowledge.org/resource/hasNumberOfPeople> ?popul1. 
  FILTER((?long, ?nlong) < 50)
} ORDER BY ASC(?popul1 + ?ecoGrowth) LIMIT k

!\bfseries{\normalsize{\textcolor{black}{Query 4:}}}!
BASE <http://yago-knowledge.org/resource/> .
PREFIX rdf: <http://www.w3.org/1999/02/22-rdf-syntax-ns#>
SELECT ?place ?nplace ?loc1 ?loc2 ?birthPlace 
WHERE {
  ?place <http://yago-knowledge.org/resource/hasPopulationDensity> ?popul. 
  ?place <http://yago-knowledge.org/resource/hasEconomicGrowth> ?ecoGrowth. 
  ?place <http://yago-knowledge.org/resource/hasNeighbor> ?birthPlace. 
  ?place <http://yago-knowledge.org/resource/hasGeometry> ?long. 
  ?place <http://yago-knowledge.org/resource/isLocatedIn> ?loc1. 
  ?nplace <http://yago-knowledge.org/resource/isLocatedIn> ?loc2. 
  ?nplace <http://yago-knowledge.org/resource/hasGeometry> ?nlong. 
  ?nplace <http://yago-knowledge.org/resource/hasNumberOfPeople> ?popul1. 
  FILTER((?long, ?nlong) < 50)
} ORDER BY ASC(?popul + ?popul1 + ?ecoGrowth) LIMIT k

!\bfseries{\normalsize{\textcolor{black}{Query 5:}}}!
BASE <http://yago-knowledge.org/resource/> .
PREFIX rdf: <http://www.w3.org/1999/02/22-rdf-syntax-ns#>
SELECT ?d?location?nplace?a?b?person 
WHERE {
  ?r rdf:subject ?b. 
  ?r rdf:predicate <http://yago-knowledge.org/resource/diedIn>. 
  ?r rdf:object ?a. 
  ?r <http://yago-knowledge.org/resource/hasConfidence> ?conf. 
  ?a <http://yago-knowledge.org/resource/isLocatedIn> ?d. 
  ?d <http://yago-knowledge.org/resource/hasGeometry> ?long. 
  ?r1 rdf:subject ?person. 
  ?r1 rdf:predicate <http://yago-knowledge.org/resource/wasBornIn>. 
  ?r1 rdf:object ?nplace. 
  ?nplace <http://yago-knowledge.org/resource/isLocatedIn> ?location. 
  ?r1 <http://yago-knowledge.org/resource/hasConfidence> ?popul1. 
  ?location <http://yago-knowledge.org/resource/hasGeometry> ?nlong. 
  FILTER((?long, ?nlong) < 50)
} ORDER BY ASC(?popul1 + ?conf) LIMIT k

!\bfseries{\normalsize{\textcolor{black}{Query 6:}}}!
BASE <http://yago-knowledge.org/resource/> .
PREFIX rdf: <http://www.w3.org/1999/02/22-rdf-syntax-ns#>
SELECT ?b?nplace?a?loc2 
WHERE {
  ?r rdf:subject ?a. 
  ?r rdf:predicate <http://yago-knowledge.org/resource/happenedIn>. 
  ?r rdf:object ?b. 
  ?r <http://yago-knowledge.org/resource/hasConfidence> ?conf. 
  ?b <http://yago-knowledge.org/resource/hasGeometry> ?long. 
  ?b <http://yago-knowledge.org/resource/hasInflation> ?d. 
  ?nplace <http://yago-knowledge.org/resource/isLocatedIn> ?loc2. 
  ?nplace <http://yago-knowledge.org/resource/hasGeometry> ?nlong. 
  ?nplace <http://yago-knowledge.org/resource/hasNumberOfPeople> ?popul1. 
  FILTER((?long, ?nlong) < 50)
} ORDER BY ASC(?d + ?popul1 + ?conf) LIMIT k

!\bfseries{\normalsize{\textcolor{black}{Query 7:}}}!
BASE <http://yago-knowledge.org/resource/> .
PREFIX rdf: <http://www.w3.org/1999/02/22-rdf-syntax-ns#>
SELECT ?a?nplace?b?loc2 
WHERE {
  ?nplace <http://yago-knowledge.org/resource/isLocatedIn> ?loc2. 
  ?nplace <http://yago-knowledge.org/resource/hasGeometry> ?nlong. 
  ?nplace <http://yago-knowledge.org/resource/hasEconomicGrowth> ?popul1. 
  ?r rdf:subject ?a. 
  ?r rdf:predicate <http://yago-knowledge.org/resource/isLocatedIn>. 
  ?r rdf:object ?b.
  ?r <http://yago-knowledge.org/resource/hasConfidence> ?conf. 
  ?a <http://yago-knowledge.org/resource/hasGeometry> ?long. 
  FILTER((?long, ?nlong) < 50)
} ORDER BY ASC(?popul1 + ?conf) LIMIT k

!\bfseries{\normalsize{\textcolor{black}{Query 8:}}}!
SELECT ?b?nplace?a?loc2 
WHERE {
  ?r1 rdf:subject ?loc2. 
  ?r1 rdf:predicate <http://yago-knowledge.org/resource/wasBornIn>. 
  ?r1 rdf:object ?nplace. 
  ?r1 <http://yago-knowledge.org/resource/hasConfidence> ?conf1. 
  ?nplace <http://yago-knowledge.org/resource/hasGeometry> ?nlong. 
  ?r rdf:subject ?a. 
  ?r rdf:predicate <http://yago-knowledge.org/resource/isLocatedIn>. 
  ?r rdf:object ?b. 
  ?r <http://yago-knowledge.org/resource/hasConfidence> ?conf. 
  ?b <http://yago-knowledge.org/resource/hasGeometry> ?long. 
  ?a <http://yago-knowledge.org/resource/hasPopulationDensity> ?d. 
  FILTER((?long, ?nlong) < 50)
} ORDER BY ASC(?d + ?conf1 + ?conf) LIMIT k

\end{lstlisting}

\end{document}